\documentclass[amsthm]{elsart}

\usepackage{yjsco}
\usepackage{natbib}
\usepackage{amsmath}
\usepackage{amssymb}
\usepackage{graphicx}
\usepackage{color}

\allowdisplaybreaks

\usepackage{bbm}

\let\set\mathbbm
\def\K{\set K}
\def\sqfp#1{#1^\ast}
\def\lc{\operatorname{lc}}
\def\deg{\operatorname{deg}}
\def\ff#1#2{#1^{\underline{#2}}}
\def\rf#1#2{#1^{\overline{#2}}}
\def\iverson#1{[\![#1]\!]}
\def\clap#1{\hbox to0pt{\hss#1\hss}}

\newtheorem{example}{Example}
\newtheorem{definition}[example]{Definition}
\newtheorem{theorem}[example]{Theorem}
\newtheorem{lemma}[example]{Lemma}
\newtheorem{corollary}[example]{Corollary}

\newtheorem{notation}[example]{Notation}

\begin{document}
\begin{frontmatter}

\title{Trading Order~for~Degree in Creative~Telescoping}

\author[1]{Shaoshi Chen}
\address{Department of Mathematics\\
  North Carolina State University\\
  Raleigh, NC 27695-8205, USA}
\ead{schen@amss.ac.cn}

\thanks[1]{Current address. The work described here was done while S.C. was employed as postdoc at RISC in the
  FWF projects Y464-N18 and P20162-N18. At NCSU, S.C. is supported by the NSF grant CCF-1017217.}

\author[2]{Manuel Kauers}
\address{Research Institute for Symbolic Computation\\
  Johannes Kepler University\\
  A4040 Linz, Austria}
\ead{mkauers@risc.jku.at}

\thanks[2]{M.K. was supported by the FWF grant Y464-N18.}

\begin{abstract}
  We analyze the differential equations produced by the method of creative
  telescoping applied to a hyperexponential term in two variables.  We show that
  equations of low order have high degree, and that higher order equations have
  lower degree.  More precisely, we derive degree bounding formulas which allow
  to estimate the degree of the output equations from creative telescoping as a
  function of the order.  As an application, we show how the knowledge of these
  formulas can be used to improve, at least in principle, the performance of creative telescoping
  implementations, and we deduce bounds on the asymptotic complexity of creative
  telescoping for hyperexponential terms.
\end{abstract}


\end{frontmatter}

\renewcommand\O{\mathrm{O}}

\section{Introduction}

Creative telescoping is a technique for computing differential or difference
equations satisfied by a given definite sum or integral. The technique became
widely known through the work of \cite{zeilberger91}, who first observed that
creative telescoping in combination with Gosper's algorithm~\citep{gosper78} for
indefinite hypergeometric summation leads to a complete algorithm for computing
recurrence equations of definite hypergeometric sums. This algorithm is now
known as Zeilberger's algorithm~\citep{zeilberger90a}. In its original version,
it accepts as input a bivariate proper hypergeometric term $f(n,k)$ and returns
as output a linear recurrence equation with polynomial coefficients satisfied by
the sum $F(n)=\sum_{k=a}^b f(n,k)$.  An analogous algorithm for definite
integration was given by \cite{almkvist90}. This algorithm accepts as input a
bivariate hyperexponential term $f(x,y)$ and returns as output a linear
differential equation with polynomial coefficients satisfied by the integral
$F(x)=\int_\alpha^\beta f(x,y) dy$. A summary of the method of creative
telescoping for this case is given in Section~\ref{sec:3} below. For further
details, variations, and generalizations, consult for instance
\cite{petkovsek97},~\cite{chyzak00}, \cite{schneider05}, \cite{chyzak09a}, \cite{kauers10j}.  For
implementations, see~\cite{paule95a}, \cite{chyzak98a}, \cite{koepf98},
\cite{schneider04b}, \cite{geddes04}, \cite{koutschan09,koutschan10b}, etc.\

The equations which can be found via creative telescoping have a certain
order~$r$ and polynomial coefficients of a certain degree~$d$. But for a fixed
integration problem, $r$~and~$d$ are not uniquely determined. Instead, there are
infinitely many points $(r,d)\in\set N^2$ such that creative telescoping can
find an equation of order~$r$ and degree~$d$. These points form a region which
is specific to the integration problem at hand. Figure~\ref{fig:1} shows an
example for such a region. Every point $(r,d)$ in the gray region
corresponds to a differential equation of order~$r$ and degree~$d$ which creative
telescoping can find for integrating the rational function
\begin{alignat*}1
  f(x,y)&=\Bigl(3 x^2 y^2+9 x^2 y+9 x^2+10 x y^2+3 x y+4x+1\Bigr)\Big/\Bigl(
     3 x^3 y^3+9 x^3 y^2+x^3 y+3 x^3\\
     &\qquad{}+7x^2 y^3+8 x^2 y^2+5 x^2+8 x y^3+10 x y^2+10 x y+x+5 y^3+10 y^2+5
     y+5\Bigr).
\end{alignat*}
The picture indicates that low order equations have high degree, and that the degree
decreases with increasing order. But what exactly is the shape of the gray
region? And where does it come from? And how can it be exploited? These are the
questions we address in this article.

\begin{figure}


  \centerline{\includegraphics{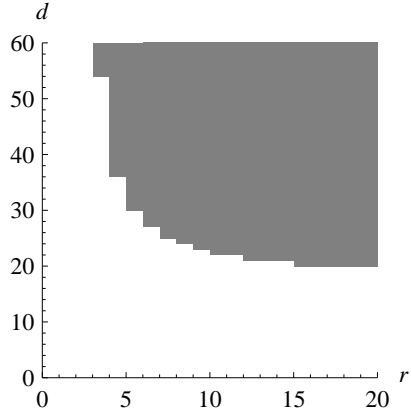}}

  \smallskip

 \caption{Sizes $(r,d)$ of creative telescoping relations for the integral of a
   certain rational function}\label{fig:1}
\end{figure}

\medskip

\textbf{How can it be exploited?}\quad There are two main reasons why the shape
of the gray region is of interest.  First, because it can be used to estimate
the size of the output equations, and hence to derive bounds on the
computational cost of computing them.  Secondly, because it can be used to
design more efficient algorithms by recognizing that some of the equations are
cheaper than others.

An analysis of this kind was first undertaken by \cite{bostan07}. They studied
the problem of computing differential equations satisfied by a given algebraic
function and found a similar phenomenon: low order equations have high degree
and vice versa. Among other things, they found that an algebraic function with a
minimal polynomial of degree~$n$ satisfies a differential equation of order at
most~$n$ with polynomial coefficients of degree $\O(n^3)$, but also a
differential equation of order~$6n$ whose coefficients have degree
only~$\O(n^2)$. Their message is that trading order for degree can pay off.

The same phenomenon applies to creative telescoping, as was shown by
\cite{bostan10b} for the case of integrating rational functions. The results in
the present article extend this work in two directions: First in that we
consider the larger input class of hyperexponential terms, and second in that we
give not only isolated degree estimates for some specific choices of~$r$, but
a curve which passes along the boundary of the gray region and thus
establishes a degree estimate as a function of the order~$r$.

\medskip

\textbf{Where does it come from?}\quad The standard argument for proving the
existence of creative telescoping relations rests on the fact that linear
systems of equations with more variables than equations must have a nontrivial
solution. Every creative telescoping relation can be viewed as a solution of a
certain linear system of equations which can be constructed from the data given
in the input. There is some freedom in how to construct these systems, and
it turns out that this freedom can be used for making the number of variables
exceed the number of equations, and thus to enforce the existence of a
nontrivial solution.

This reasoning not only implies the existence of equations and the termination
of the algorithm which searches for them, but it also implies bounds on the
output size and on the computational cost of the algorithm. But in order
to obtain good bounds, the freedom in setting up the linear systems must be used
carefully. For a good bound, we not only want that the number of variables
exceeds the number of equations, but we also want this to happen already for a
reasonably small system. The shape of the gray region originates from the smallest
systems which have solutions.

\cite{verbaeten74,verbaeten76} introduced a technique which helps in keeping the
size of the systems small. The idea is to saturate the linear systems by
introducing additional variables in a way that avoids increasing the number of
equations. We will make use of this idea in Section~\ref{sec:4} where we propose
a design for a parameterized family of linear systems whose solutions give rise
to creative telescoping relations. Unfortunately, it requires some quite lengthy
and technical calculations to translate this particular design into an
inequality condition which rephrases the condition ``number of variables $>$
number of equations'' in precise terms. However, as a reward we obtain a good
approximation to the gray region as the solution of this inequality.

\medskip

\textbf{What is the exact shape?}\quad We don't know. All we can offer are
some rational functions which describe the boundary of the region of all $(r,d)$
where the ansatz described in Section~\ref{sec:4} has a solution
(Theorem~\ref{thm:1}). The graphs of these rational functions are curves which
pass approximately along the boundary of the gray region.

By construction, for all integer points $(r,d)$ above these graphs we can
guarantee the existence of a creative telescoping relation of order~$r$ with
polynomial coefficients of degree~$d$. But we have no proof that our curves are
best possible. Experiments have shown that at least in some cases, our curve
describes the boundary of the gray region exactly, or within a negligible
error. In other cases, there remains a significant portion of the gray region
below our curve when $r$ is large.

In cases where the curve from Theorem~\ref{thm:1} is tight, we can compute the
points $(r,d)$ for which certain interesting measures (such as computing time,
output size, \dots) are minimized, as shown in Section~\ref{sec:6}. Even when
the curve is not tight, these calculations still give rise to new asymptotic
bounds (including the multiplicative constants) of the corresponding
complexities. We expect that this data will be valuable for
constructing the next generation of symbolic integration software.

\section{Creative Telescoping for Hyperexponential Terms}\label{sec:3}

We consider in this article only hyperexponential terms as integrands.
Throughout the article, $\K$~is a field of characteristic~$0$,
and $\K(x,y)$ is the field of bivariate rational functions in $x$ and~$y$ over~$\K$.
Let~$D_x$ and~$D_y$ denote the derivations on~$\K(x, y)$ such that~$D_xc=D_yc=0$
for all~$c\in \K$, and~$D_xx=1$, $D_xy=0$, $D_yx=0$, $D_yy=1$.
One can see that~$D_x$ and~$D_y$
commute with each other on~$\K(x, y)$. We say that a field~$\set E$
containing~$\K(x, y)$ is a \emph{differential field extension}
of $\K(x,y)$ if the derivations~$D_x$ and~$D_y$ are extended to derivations on~$\set E$
and those extended derivations, still denoted by~$D_x$ and~$D_y$, commute with each other on~$\set E$.

\begin{definition}
 An element $h$ of a differential field extension $\set E$ of $\K(x,y)$
  is called \emph{hyperexponential} (over $\K(x,y)$) if
  \[
    \frac{D_x h}{h}\in\K(x,y)\quad\text{and}\quad
    \frac{D_y h}{h}\in\K(x,y).
  \]
\end{definition}

When $h\in\set E$ is a hyperexponential term and $r_1,r_2\in\K(x,y)$ are such that
$(D_x h)/h=r_1$ and $(D_y h)/h=r_2$, then $D_xD_yh=D_yD_xh$ implies $D_yr_1=D_xr_2$.
Conversely, \cite{christopher99} has shown for algebraically closed ground fields $\K$ that
for any two rational functions $r_1,r_2\in\K(x,y)$ with $D_yr_1=D_xr_2$
there exist $a/b\in\K(x,y)$,
$c_0,\dots,c_L\in\K[x,y]$ and $e_1,\dots,e_L\in\K$ with
\[
 r_1=\frac{D_xc_0}{c_0} + D_x\Bigl(\frac ab\Bigr) + \sum_{\ell=1}^L e_\ell \frac{D_xc_\ell}{c_\ell}
 \quad\text{and}\quad
 r_2=\frac{D_yc_0}{c_0} + D_y\Bigl(\frac ab\Bigr) + \sum_{\ell=1}^L e_\ell \frac{D_yc_\ell}{c_\ell}.
\]
Together with Theorem~2 of \cite{bronstein05}, it follows that there exists
a differential field extension~$\set E$ of $\K(x,y)$ and an element
$h\in\set E$ with $(D_xh)/h=r_1$ and $(D_yh)/h=r_2$ which we can write
in the form
\[
  h = c_0\exp\Bigl(\frac{a}{b}\Bigr)\prod_{\ell=1}^L c_\ell^{e_\ell},
\]
where $a\in\K[x,y]$, $b,c_0,\dots,c_L\in\K[x,y]\setminus\{0\}$,
$e_1,\dots,e_\ell\in\set K$, and the expressions $\exp(a/b)$ and $c_\ell^{e_\ell}$
refer to elements of $\set E$ on which $D_x$ and $D_y$ act as suggested by
the notation. We assume from now on that hyperexponential terms are always
given in this form, and we use the letters $a,b,c_0,\dots,c_L,e_1,\dots,e_L$
consistently throughout with the meaning they have here.

\begin{example}
 $h=\exp(x^2y)\sqrt{x-2y}$ is a hyperexponential term. We have
  \begin{alignat*}1
   \frac{D_x h}h &= \frac{1+4x^2y-8xy^2}{2x-4y} = 2xy + \frac1{2x-4y}\in\K(x,y),\\
   \frac{D_y h}h &= \frac{x^3-2x^2y-1}{x-2y} = x^2 - \frac1{x-2y} \in\K(x,y).
  \end{alignat*}
  For this term, we can take $c_0=1$, $a=x^2y$, $b=1$, $c_1=x-2y$, $e_1=\tfrac12$.
\end{example}

We may adopt the
additional condition (without loss of generality) that the $c_\ell$ ($\ell>0$)
are square free and pairwise coprime, and that $e_\ell\not\in\set N$ for all
$\ell>0$. The estimates derived below do not depend on these additional
conditions, but will typically not be sharp when they are not fulfilled. For
simplicity, we will exclude throughout some trivial special cases by assuming
that all $e_\ell$ are nonzero and that
$\max\{\deg_xa,\deg_xb\}+\sum_{\ell=1}^L\deg_xc_\ell$ and
$\max\{\deg_ya,\deg_yb\}+\sum_{\ell=1}^L\deg_yc_\ell$ are nonzero. These latter
two conditions encode the requirement that $h$ is neither independent of $x$ nor
independent of~$y$, nor simply a polynomial.

Applied to the hyperexponential term~$h$, the method of creative telescoping
consists of finding, by whatever means, polynomials $p_0,\dots,p_r\in\K[x]$, not
all zero, and a hyperexponential term~$Q$ such that
\[
  p_0 h
  +p_1 D_x h
  +\cdots
  +p_r D_x^r h
  = D_y Q.
\]
An equation of this form is called a \emph{creative telescoping relation}
for~$h$, the differential operator $P:=p_0+p_1D_x+\cdots+p_rD_x^r$
appearing on the left is called the \emph{telescoper} and $Q$ is called the
\emph{certificate} of the relation. The telescoper is required to be nonzero and
free of~$y$, but the certificate may be zero or it may involve both $x$ and~$y$.
When $p_r\neq0$, the number~$r$ is called the \emph{order} of~$P$, and
$d:=\max_{i=0}^r \deg_x p_i$ is called its \emph{degree}.

To motivate the form of a creative telescoping relation, assume that $h=h(x,y)$ can be interpreted
as an actual function in $x$ and $y$ and consider the integral $f(x)=\int_\alpha^\beta h(x,y)dy$.
Then integrating both sides of a creative telescoping relation implies that $f$
satisfies the inhomogeneous differential equation
\[
  p_0(x) f(x)
  +p_1(x) D_x f(x)
  +\cdots
  +p_r(x) D_x^r f(x)
  = \bigl[Q(x,y) \bigr]_{y=\alpha}^\beta.
\]
In the frequent situation that the inhomogeneous part happens to evaluate to zero,
this means that the telescoper of $h$ annihilates the integral~$f$.

\begin{example}
  A creative telescoping relation for $h=\exp(x^2y)\sqrt{x-2y}$ is
  \[
    (3x^3-6)h-2x D_x h = D_y\bigl((3x-4y) h\bigr).
  \]
  It consists of the telescoper $P=(3x^3-6)-2x D_x$ and the certificate $Q=(3x-4y)h$.
  For the definite integral $f(x) := \int_{-\infty}^{x/2} \exp(x^2y)\sqrt{x-2y} dy$,
  we obtain the differential equation
  \begin{alignat*}1
    (3x^3-6)f(x)-2x D_x f(x) &=0.
  \end{alignat*}
\end{example}

Creative telescoping relations for hyperexponential terms can be found with the
algorithm of~\cite{almkvist90}, which relies on a continuous analogue of Gosper's summation
algorithm. A more direct approach was considered by Apagodu (alias Mohammed) and
Zeilberger~\citeyearpar{mohammed05,apagodu06}. After making a suitable choice
for the denominator of~$Q$, they fix an order~$r$ and a degree~$s$
for the numerator of~$Q$, make an ansatz with undetermined
coefficients, and obtain a linear system by comparing coefficients. Appropriate
choices of $r$ and $s$ ensure that this linear system has a nontrivial
solution, and also lead to a sharp bound on the order~$r$ of the telescoper.

Let us illustrate this reasoning for the case where the integrand is a rational
function $h=u/v\in\K(x,y)$ with $\deg_y u<\deg_y v$ and $v$ irreducible.  Fix
some~$r$. Then we have to find $p_0,\dots,p_r\in\K(x)$ and a rational function
$Q\in\K(x,y)$ with
\[
  p_0 h + p_1 D_x h + \cdots + p_r D_x^r h = D_y Q.
\]
A reasonable choice for $Q$ is $Q=\bigl(\sum_{i=0}^s q_i y^i\bigr)/
v^r$, where $s=\deg_y u + (r-1) \deg_y v$ and
$q_0,\dots,q_s$ are unknowns, because with this choice, both sides of
the equation are equal to a rational function with the same denominator
$v^{r+1}$ and numerators of degree at most $\deg_y u + r \deg_y v$
in~$y$ in which the unknowns $p_i$ and $q_j$ appear linearly.  Comparing
coefficients with respect to $y$ on both sides leads to a homogeneous linear
system of at most $1+\deg_y u + r \deg_y v$ equations with $(r+1)+(s+1)$
unknowns and coefficients in~$\K(x)$. This system will have a nontrivial
solution if $r$ is chosen such that
\[
  (r+1)+(s+1) > \deg_y u + r \deg_y v + 1
  \iff
  r \geq \deg_y v.
\]
All these solutions must lead to a nonzero telescoper~$P$ because any
nontrivial solution with $P=0$ would have a nonzero certificate $Q$ with
$D_y Q=0$, and this is impossible because $s$ was chosen such that the
numerator of $Q$ has a strictly lower degree than its denominator.

We have thus shown the existence of telescopers of any order $r\geq\deg_y
v$. This is a good bound, but it does not provide any estimate on their
degrees~$d$. We will next derive inequalities involving both $r$ and~$d$ by
constructing linear systems with coefficients in $\K$ rather than in~$\K(x)$.

\section{Shaping the Ansatz}\label{sec:4}\label{sec:ansatz}

Let $h$ be a hyperexponential term and consider an ansatz of the form
\[
  P = \sum_{i=0}^r \sum_{j=0}^{d_i} p_{i,j}x^j D_x^i, \qquad
  Q = \biggl(\sum_{i=0}^{s_1}\sum_{j=0}^{s_2} q_{i,j}x^i y^j\biggr)\frac{h}{v}
\]
for a telescoper~$P$ and a certificate~$Q$. The plan is to find a good choice
for the parameters $r,s_1,s_2,v,d_0,\dots,d_r$. The only restriction we have is
that the linear system obtained from equating all the coefficients in the
numerator of the rational function $(Ph-D_yQ)/h$ to zero should have a solution
in which not all the $p_{i,j}$ are zero. The remaining freedom can be used to
shape the ansatz such as to keep $d:=\max_{i=0}^r d_i$ small.

As a sufficient condition for the existence of a solution, we will require that
the number of terms $x^i y^j$ in the numerator of the rational function
$(Ph-D_yQ)/h$ (i.e., the number of equations) should be less than
$\sum_{i=0}^r(d_i+1)+(s_1+1)(s_2+1)$ (i.e., the number of variables $p_{i,j}$
and~$q_{i,j}$). As shown in the following example, this condition is really
just sufficient, but not necessary.

\begin{example} Let $h=u/v$ be the rational function from the introduction.
  With $r=3$, $d_0=d_1=d_2=d_3=d=54$, and $Q=\bigl(\sum_{i=0}^{62}\sum_{j=0}^{8}
  q_{i,j}x^i y^j\bigr)\big/v^3$, comparing the coefficients of the numerator of
  $(Ph-D_yQ)/h$ to zero gives a linear system with $787$ variables and $792$
  equations. This system has a nonzero solution although $792>787$.
\end{example}

This phenomenon is not only an unlikely coincidence in this particular example,
but it happens systematically when the parameters of the ansatz are not well
chosen. Estimates which are only based on balancing the number of variables and
the number of equations will then overshoot. It is therefore preferable to shape
the ansatz for $P$ and $Q$ in such a way that the linear system originating from
it will have a nullspace whose dimension is exactly the difference between the
number of equations and the number of variables (or 0 if there are more
equations than variables).

The goal of this section is to describe our choice for the ansatz of telescoper
and certificate. The form of the ansatz for the telescoper is given in
Section~\ref{sec:3.1}, the certificate is discussed in Section~\ref{sec:3.2}.
In the beginning, we collect some facts about the rational functions
$(D_x^ih)/h$ which are used later for calculating how many equations a
particular ansatz induces. The following notational conventions will be
used throughout.

\begin{notation}
\begin{itemize}
\item $\lc_z p$ and $\deg_z p$ refer to the leading coefficient and the degree
  of the polynomial~$p$ with respect to the variable~$z$, respectively.  For the
  zero polynomial, we define $\deg_z0:=-\infty$ and $\lc_z0:=0$.
\item $\sqfp p$ refers to the square free part of the polynomial~$p$ with respect to
  all its variables, e.g., $\sqfp{\bigl((x+1)^3 (y+3)^2\bigr)}=(x+1)(y+3)$.
  Note that $\sqfp p$ is only unique up to multiplication by elements from $\K\setminus\{0\}$,
  but that for any choice of $\sqfp p$, the degrees $\deg_x \sqfp p$ and $\deg_y \sqfp p$
  are uniquely determined and we have that $\sqfp p (D_x p)/p$ is a polynomial in $x$ and~$y$.
  These are the only properties we will use.
\item $\ff zn:=z(z-1)(z-2)\cdots(z-n+1)$ and $\rf zn:=z(z+1)(z+2)\cdots(z+n-1)$ denote
  the falling and rising factorials, respectively. For $n\leq0$ we define $\ff zn:=\rf zn:=1$.
\item If $z$ is a real number, then $z^+:=\max\{0,z\}$.
\item If $z$ is a real number, then $\lfloor z\rfloor:=\max\{x\in\set Z:x\leq z\}$,
  $\lceil z\rceil:=\min\{x\in\set Z:x\geq z\}$, and
  $\lfloor z\rceil:=\lfloor z+\tfrac12\rfloor$ denotes the nearest integer to~$z$.
\item If $\Phi$ is a formula then $\iverson\Phi$ denotes the Iverson bracket, which
  evaluates to $1$ if $\Phi$ is true and to $0$ if $\Phi$ is false, e.g.,
  $z^+=\iverson{z\geq0}z$; $\delta_{i,j}=\iverson{i=j}$, etc.
\end{itemize}
\end{notation}

\begin{lemma}\label{lemma:7} Let $h$ be a hyperexponential term and $i\geq0$.
\begin{enumerate}
\item\label{lemma:7:1} If $\deg_x a > \deg_x b$, then
  \[
   \frac{D_x^i h}{h} = \frac{N_i}{c_0\bigl(b\sqfp  b\prod_{\ell=1}^L c_\ell\bigr)^i}
  \]
  for some polynomial $N_i\in\K[x,y]$ with
  \begin{alignat*}1
   \deg_x N_i&=\deg_xc_0 + i\Bigl(\deg_x a+\deg_x \sqfp b+\sum_{\ell=1}^L \deg_x c_\ell - 1\Bigr),\\[-3pt]
   \deg_y N_i&\leq \deg_yc_0+i\Bigl(\max\{\deg_y a,\deg_y b\} + \deg_y \sqfp b + \sum_{\ell=1}^L \deg_y c_\ell\Bigr),\\[-3pt]
   \lc_x N_i &=(\lc_x c_0)\Bigl(\lc_x a\sqfp b\prod_{\ell=1}^L c_\ell\Bigr)^i(\deg_xa-\deg_xb)^i.
  \end{alignat*}
\item\label{lemma:7:2} If $\deg_x a\leq\deg_x b$, then
  \[
   \frac{D_x^i h}{h} = \frac{N_i}
                  {c_0\bigl(b\sqfp b\prod_{\ell=1}^L c_\ell\bigr)^i}
  \]
  for some polynomial $N_i\in\K[x,y]$ with
  \begin{alignat*}1
    \deg_x N_i&=\deg_x c_0 + i\Bigl(\deg_x b+\deg_x\sqfp b+\sum_{\ell=1}^L \deg_x c_\ell - 1\Bigr)
               - \iverson{\omega\in\set N\land i>\omega}\delta,\\[-3pt]
    \deg_y N_i&\leq \deg_yc_0+ i\Bigl(\max\{\deg_y a,\deg_y b\} + \deg_y \sqfp b + \sum_{\ell=1}^L \deg_y c_\ell\Bigr),\\[-3pt]
     \lc_x N_i&=\left\{\begin{array}{ll}
           (\lc_x c_0)\bigl(\lc_xb\sqfp b\prod\limits_{\ell=1}^Lc_\ell\bigr)^i
        \ff{\omega}i&\quad\text{if $\omega\not\in\set N$ or $i\leq\omega$;}\\
           (\lc_x N_{\omega+1})
           \bigl(\lc_x b\sqfp b\prod\limits_{\ell=1}^Lc_\ell\bigr)^{i-(\omega+1)}
           \ff{(-\delta-1)}{i-(\omega+1)}
           &\quad\text{if $\omega\in\set N$ and $i>\omega$,}
         \end{array}\right.
   \end{alignat*}
   where $\omega:=\deg_xc_0+\sum\limits_{\ell=1}^L e_\ell\deg_x c_\ell$ and,
   if $\omega\in\set N$,
   \[
    \delta:=\deg_xc_0+(\omega+1)\Bigl(\deg_x b+\deg_x\sqfp b
     +\sum_{\ell=1}^L\deg_x c_\ell-1\Bigr)-\deg_x N_{\omega+1}\geq1.
   \]
\end{enumerate}
\end{lemma}
\begin{proof} All claims are proved by induction on~$i$. For $i=0$, there is nothing
  to show in any of the cases. The calculations for the induction step $i\to i+1$ are
  as follows.
  \begin{enumerate}
  \item Let $v:=b\sqfp b\prod_{\ell=1}^L c_\ell$ and write $m_i$ for the claimed
    value of $\deg_xN_i$. Then
    \begin{alignat*}1
      \frac{D_x^{i+1}h}{h}&=D_x\biggl(\frac{N_i}{c_0v^i}
        c_0\exp\Bigl(\frac{a}{b}\Bigr)\prod_{\ell=1}^Lc_\ell^{e_\ell}
      \biggr)\bigg/\biggl(c_0\exp\Bigl(\frac{a}{b}\Bigr)\prod_{\ell=1}^Lc_\ell^{e_\ell}\biggr)\\
      &=\frac{(D_xN_i)v-iN_iD_xv}{c_0v^{i+1}}+\frac{N_i}{c_0v^i}\frac{(D_xa)\sqfp b-a\sqfp b(D_xb)/b}{b\sqfp b}
        +\frac{N_i}{c_0v^i}\sum_{\ell=1}^Le_\ell\frac{D_xc_\ell}{c_\ell}\\
      &=\frac{(D_xN_i)v-iN_iD_xv+N_i\bigl(\prod\limits_{\ell=1}^Lc_\ell\bigr)\bigl((D_x
        a)\sqfp b-a\sqfp b\frac{D_xb}b\bigr)+N_iv\sum\limits_{\ell=1}^L e_\ell\frac{D_x
          c_\ell}{c_\ell}}{c_0v^{i+1}}.
    \end{alignat*}
    Since $\deg_x a>\deg_x b$ by assumption, we have
    \begin{alignat*}1
     &\deg_x\biggl((D_xN_i)v-i N_iD_xv+N_i v\sum_{\ell=1}^L e_\ell\frac{D_xc_\ell}{c_\ell}\biggr)\\
     &\qquad{}\leq\deg_x N_i+\deg_x v - 1
        = m_i + \deg_x\sqfp b + \deg_x b + \smash{\sum_{\ell=1}^L}\deg_xc_\ell - 1\\
     &\qquad{}< m_i + \deg_x a + \deg_x\sqfp b + \sum_{\ell=1}^L\deg_xc_\ell - 1
      = m_{i+1}.
    \end{alignat*}
    Furthermore, because of
    \[
      (D_x a)\sqfp b - a\sqfp b\frac{D_x b}b =
       (\lc_x a)(\lc_x \sqfp b)(\deg_x a-\deg_x b)x^{\deg_x a+\deg_x\sqfp  b-1} + \cdots
    \]
    we have
    \begin{alignat*}1
        N_i\Bigl((D_xa)\sqfp b-a\sqfp b\frac{D_xb}b\Bigr)\prod_{\ell=1}^L c_\ell
     &= (\lc_xN_i)\Bigl(\lc_x a\sqfp b\prod_{\ell=1}^Lc_\ell\Bigr)(\deg_x a-\deg_x b)x^{m_{i+1}}+\cdots.
    \end{alignat*}
    This completes the proof that $(D_x^{i+1}h)/h$ has the denominator as claimed and that its
    numerator has degree and leading coefficient with respect to $x$ as claimed.
    The remaining degree bound with respect to $y$ follows from
    \begin{alignat*}1
          &\deg_y\Bigl((D_xN_i)v-i N_iD_xv+
          N_i v\sum_{\ell=1}^L e_\ell\frac{D_xc_\ell}{c_\ell}\Bigr)\\
      \leq{}&
        \underbrace{\deg_yc_0{+}i\Bigl(\deg_y\sqfp b{+}\max\{\deg_y a,\deg_y b\}{+}\!\sum_{\ell=1}^L \deg_y c_\ell\Bigr)}
        _{\text{bounds $\deg_y N_i$}}
         +
        \underbrace{\deg_y \sqfp b{+}\deg_y b{+}\!\sum_{\ell=1}^L \deg_y c_\ell}
        _{\text{bounds $\deg_y v$}}\\
      \leq{}&\deg_yc_0+(i+1)\Bigl(\deg_y\sqfp b + \max\{\deg_y a,\deg_y b\} + \sum_{\ell=1}^L \deg_y c_\ell\Bigr)
    \end{alignat*}
    and
    \begin{alignat*}1
        &\deg_y\Bigl(N_i\Bigl((D_xa)\sqfp b-a\sqfp b\frac{D_xb}b\Bigr)\prod\limits_{\ell=1}^Lc_\ell\Bigr)\\
     \leq{}&
        \underbrace{\deg_yc_0{+}i\Bigl(\deg_y\sqfp b{+}\max\{\deg_y a,\deg_y b\}{+}\!\sum_{\ell=1}^L \deg_y c_\ell\Bigr)}
        _{\text{bounds $\deg_y N_i$}}
        +
        \underbrace{\deg_y\sqfp b {+} \deg_y a{+}\sum_{\ell=1}^L \deg_y c_\ell}
        _{\text{bounds $\deg_y$ of the other factors}}\\
     \leq{}&\deg_yc_0+(i+1)\Bigl(\deg_y\sqfp b + \max\{\deg_y a,\deg_y b\} + \sum_{\ell=1}^L \deg_y c_\ell\Bigr).
    \end{alignat*}
  \item Again, let $v:=  b\sqfp b\prod_{\ell=1}^L c_\ell$
    and write $m_i$ for the claimed value of $\deg_x N_i$.
    Then, like in part~\ref{lemma:7:1},
    \[
      \frac{D_x^{i+1}h}{h}=\frac{(D_xN_i)v-i N_iD_xv+ N_i\bigl(\prod\limits_{\ell=1}^Lc_\ell\bigr)\bigl((D_x
           a)\sqfp b-a\sqfp b\frac{D_x b}b\bigr) +N_iv\sum\limits_{\ell=1}^L e_\ell\frac{D_x
            c_\ell}{c_\ell}}{c_0v^{i+1}}.
    \]
    First consider the case $\omega\not\in\set N$ or $i\leq\omega$.

    Since $\deg_x a\leq\deg_x b$ by assumption, and because of
    \[
      (D_x a)\sqfp b-a\sqfp b\frac{D_x b}b=(\lc_x a)(\lc_x \sqfp b)(\deg_x a-\deg_x b)x^{\deg_x a+\deg_x\sqfp b-1}+\cdots,
    \]
    we now have
    \begin{alignat*}1
      &\deg_x \Bigl(N_i\Bigl((D_x a)\sqfp b-a\sqfp b\frac{D_x b}b\Bigr)\prod_{\ell=1}^L c_\ell\Bigr)\\
      &\quad{}< m_i + \deg_x b+\deg_x\sqfp b-1 + \sum_{\ell=1}^L \deg_x c_\ell=m_{i+1}.
    \end{alignat*}
    Note that this estimate is also strict when $\deg_x a=\deg_x b$ because
    the coefficient of $x^{\deg_x a+\deg_x\sqfp b-1}$ in $(D_x a)\sqfp b-a\sqfp b (D_x b)/b$
    contains the factor $\deg_x a-\deg_x b$, which vanishes in this case.

    Next, using the induction hypothesis, we have
    \begin{alignat*}1
        &(D_x N_i)v-i N_iD_x v+ N_i v\sum_{\ell=1}^L e_\ell\frac{D_xc_\ell}{c_\ell}\\
     ={}&(\lc_x N_i)(\lc_x v)\Bigl(\deg_x N_i- i\deg_x v + \sum_{\ell=1}^L e_\ell\deg_x c_\ell\Bigr)x^{\deg_x N_i+\deg_x v-1}+\cdots\\
     ={}&(\lc_x c_0)\Bigl(\lc_x b\sqfp b\prod_{\ell=1}^Lc_\ell\Bigr)^i
     \ff{\omega}i\,(\lc_x v)
           \Bigl(m_i - i\deg_x v + \sum_{\ell=1}^L e_\ell\deg_x c_\ell\Bigr)
           x^{m_i+\deg_x v-1}+\cdots\\
     ={}&(\lc_x c_0)\Bigl(\lc_x b\sqfp b\prod_{\ell=1}^Lc_\ell\Bigr)^{i+1}
           \ff{\omega}i
           \Bigl(\deg_xc_0+\sum_{\ell=1}^L e_\ell\deg_x c_\ell - i \Bigr)
           x^{m_{i+1}}+\cdots\\
     ={}&(\lc_x c_0)\Bigl(\lc_x b\sqfp b\prod_{\ell=1}^Lc_\ell\Bigr)^{i+1}
         \ff{\omega}{i+1}
           x^{m_{i+1}}+\cdots.
    \end{alignat*}
    Since $\ff\omega {i+1}\neq0$ when $\omega\not\in\set N$ or $i+1\leq\omega$, this completes
    the proof that $(D_x^{i+1}h)/h$ has the denominator as claimed and that its
    numerator has degree and leading coefficient with respect to $x$ as
    claimed. The degree bounds with respect to $y$ are shown exactly as in
    part~\ref{lemma:7:1}.

    Now consider the case where $\omega\in\set N$ and $i>\omega$. In this
    case, we start the induction at $i=\omega+1$. The induction base follows
    from the calculations carried out above for $i\geq\omega$, the fact
    $\ff\omega{\omega+1}=0$, and the definition of~$\delta$.
    (Note that $\ff\omega{\omega+1}=0$ also implies $\delta\geq1$.)
    For the induction step $i\mapsto i+1$, we have, similar as before,
    \[
     \deg_x\Bigl(N_i\Bigl((D_xa)\sqfp b-a\sqfp b\frac{D_xb}{b}\Bigr)\prod_{\ell=1}^Lc_\ell\Bigr)
        < m_{i+1}
    \]
    and
    \begin{alignat*}1
         &(D_xN_i)v-iN_iD_xv+N_iv\sum_{\ell=1}^Le_\ell\frac{D_xc_\ell}{c_\ell}\\
      ={}&(\lc_xN_i)(\lc_x v)\Bigl(m_i-i\deg_xv+\sum_{\ell=1}^L e_\ell\deg_x c_\ell\Bigr)x^{m_i+\deg_x v-1}+\cdots\\
      ={}&(\lc_xN_i)(\lc_x v)\Bigl(\deg_xc_0 +i(\deg_xv-1)-\delta-i\deg_xv+\sum_{\ell=1}^L e_\ell\deg_x c_\ell\Bigr)x^{m_{i+1}}+\cdots\\
      ={}&(\lc_x N_{\omega+1})
      \Bigl(\lc_x b\sqfp b\prod_{\ell=1}^L c_\ell\Bigr)^{i-(\omega+1)}
      \ff{(-\delta-1)}{i-(\omega+1)}
(\lc_x v)(\omega-\delta-i)x^{m_{i+1}}+\cdots\\
      ={}&(\lc_x N_{\omega+1})\Bigl(\lc_x b\sqfp b\prod_{\ell=1}^L c_\ell\Bigr)^{i+1-(\omega+1)}
         \ff{(-\delta-1)}{i+1-(\omega+1)}x^{m_{i+1}}+\cdots
    \end{alignat*}
    Because of $\delta>0$, the factor $\ff{(-\delta-1)}{i-(\omega+1)}$ is nonzero for all $i>\omega$.
    \qed
\end{enumerate}
\let\qed\empty 
\end{proof}

\begin{example} The case when $h=u/v$ is a rational function is covered by
  part~\ref{lemma:7:2} of the lemma.
  For example, for $h=(2x^5-3x^4+5)/(3x^3-4x+8)$ we
  can take $c_0=2x^5-3x^4+5$, $a=0$, $b=1$, $L=1$,
  $c_1=3x^3-4x+8$, $e_1=-1$. Direct calculation of the derivatives gives
  \begin{center}
    \medskip
    \begin{tabular}{c|c|c|c|c|c|c|c}
      $i$ & 0 & 1 & 2 & 3 & 4 & 5 & 6 \\\hline
      $\deg_x c_0c_1^i (D_x^i h)/h$ & 5 & 7 & 9 & 8 & 10 & 12 & 14 \\\hline
      $\lc_x c_0c_1^i (D_x^i h)/h$ & 2 & 12 & 36 & 1512 & $-18144$ & 272160 & $-4898880$
    \end{tabular}
    \medskip
  \end{center}
  The lemma makes no statement about the degree or leading coefficient
  in the case $i=\omega+1=3$, but knowing these, it correctly predicts
  all the other data in the table. In this example, we have $\delta=3=\omega+1$.
  This is not a coincidence, as we shall show next.
\end{example}

\begin{lemma}\label{lemma:delta}
  Let $h$ be a hyperexponential term with $\deg_x a\leq\deg_x b$, and let
  $\omega$ and $\delta$ be as in Lemma~\ref{lemma:7}.(\ref{lemma:7:2}),
  $\omega\in\set N$. Then $\delta\geq\omega+1$.
\end{lemma}
\begin{proof}
  Rewrite $h=c_0\exp(\frac ab)\prod_{\ell=1}^L c_\ell^{e_\ell}= \bar
  c_0\exp(\frac ab)\prod_{\ell=1}^{L+2}\bar c_\ell^{\bar e_\ell}$ with $\bar
  c_0=x^\omega$, $\bar c_\ell=c_\ell$ ($\ell=1,\dots,L$), $\bar e_\ell=e_\ell$
  ($\ell=1,\dots,L$), $\bar c_{L+1}=c_0$, $\bar e_{L+1}=1$, $\bar c_{L+2}=x$,
  $\bar e_{L+2}=-\omega$.  The rational functions $(D_x^ih)/h$ are of course
  independent of the representation of~$h$, but the representations of these
  rational functions which are given in Lemma~\ref{lemma:7} are not. The
  representation obtained for the new representation of $h$ is obtained from the
  original representation by multiplying numerator and denominator by
  $x^{\omega+i}c_0^{i-1}$.  Observe that this modification does not influence
  the values for $\omega$ or~$\delta$.  It is therefore sufficient to prove the
  claim for terms of the form $h=x^\omega\bar h$, where $\bar h$ is some hyperexponential
  term for which the value of $\omega$ is zero.  We do so by induction
  on~$\omega$.  For $\omega=0$, we have $\delta\geq 1=\omega+1$ already
  by Lemma~\ref{lemma:7}.(\ref{lemma:7:2}).  Now assume that $\omega\geq0$
  is such that for $x^\omega\bar h$ the degree drop $\bar\delta$ is $\omega+1$
  or more.  Then for $h=x^{\omega+1}\bar h=x(x^\omega\bar h)$ we have
  $D_xh=x^\omega\bar h + xD_x(x^\omega\bar h)$, $D_x^2h=2D_x(x^\omega\bar h) +
  xD_x^2(x^\omega\bar h)$, and so on, all the way down to
  \begin{alignat}1
   D_x^{\omega+2}h&=(\omega+2)D_x^{\omega+1}(x^\omega\bar h) + x D_x^{\omega+2}(x^\omega\bar h)\notag\\
    &=(\omega+2)\frac{N_{\omega+1}}{x^\omega v^{\omega+1}}x^\omega\bar h + x \frac{N_{\omega+2}}{x^\omega v^{\omega+2}}x^\omega\bar h\notag \\
    &=\frac{(\omega+2) N_{\omega+1}v + x N_{\omega+2}}{v^{\omega+2}}\bar h,
    \label{eq:5}
  \end{alignat}
  where $N_{\omega+1}$ and $N_{\omega+2}$ are as in Lemma~\ref{lemma:7} and
  $v$ refers to the denominator stated there.
  If $\delta$ denotes the degree drop for~$h$, then this calculation implies
  $\delta\geq\bar\delta$.
  By induction hypothesis, we have $\bar\delta\geq\omega+1$. If in fact $\bar\delta\geq\omega+2$,
  then we are done. Otherwise, if $\bar\delta=\omega+1$, then
  \[
    \lc_x N_{\omega+2} = (-\bar\delta - 1)\lc_x N_{\omega+1} \lc_x v
                       = -(\omega+2)\lc_x N_{\omega+1}\lc_x v
  \]
  by Lemma~\ref{lemma:7}, so the leading terms of the two polynomials in the numerator of~\eqref{eq:5}
  cancel, and therefore $\delta>\omega+1$ also in this case.
\end{proof}

Experiments suggest that the bound in Lemma~\ref{lemma:delta} is tight in the
sense that we have $\delta=\omega+1$ for almost all hyperexponential terms~$h$.
But there do exist situations with $\delta>\omega+1$. For example, it can
be shown that for $h=c_0\exp(a/b)$ with $\deg_xb-\deg_xa>\deg_xc_0=\omega$ we
have $\delta\geq\deg_xb-\deg_xa$.

Also Lemma~\ref{lemma:7} is not necessarily sharp for degenerate choices
of~$h$.  In particular, we do not claim that the numerators and denominators
stated in Lemma~\ref{lemma:7} are coprime. It may be possible to carry out a
finer analysis by considering the square free decomposition of~$c_0$, or by
taking into account possible common factors between $b$ and the~$c_\ell$, or by
handling the $c_\ell$ which do not involve $x$ separately. For our purpose, we
believe that the statements given above form a reasonable compromise between
sharpness of the statements and readability of the derivation.

Several aspects of the formulas in Lemma~\ref{lemma:7} are important. One of
them is that the denominators corresponding to lower derivatives divide
those corresponding to higher derivatives. This has the consequence that
when the linear combination $Ph$ is brought on a common denominator, the degree
of the numerator will not grow drastically. In a sense, this fact is the main
reason why creative telescoping works at all. Our next step is to bring the
formulas from Lemma~\ref{lemma:7} on a common denominator.

\begin{lemma}\label{lemma:8} Let $h$ be a hyperexponential term and
  $r,i\in\set Z$ with $r\geq i\geq 0$.
\begin{enumerate}
\item\label{lemma:8:1} If $\deg_x a > \deg_x b$, then
  \[
   \frac{D_x^i h}h = \frac{N_{r,i}}{c_0\bigl( b\sqfp b\prod_{\ell=1}^L c_\ell\bigr)^r}
  \]
  for some $N_{r,i}\in\K[x,y]$ with
  \begin{alignat*}1
    \deg_x N_{r,i}&=
      \deg_xc_0+r\Bigl(\deg_x\sqfp b+\deg_x b+\sum_{\ell=1}^L\deg_xc_\ell\Bigr)+i\Bigl(\deg_x a - \deg_x b - 1\Bigr)
    \\[-3pt]
    \deg_y N_{r,i}&\leq
    \deg_yc_0+r\Bigl(\deg_y\sqfp b + \max\{\deg_y a,\deg_y b\} + \sum_{\ell=1}^L \deg_y c_\ell\Bigr),
    \\[-3pt]
    \lc_x N_{r,i}&=(\lc_xc_0)(\lc_x a)^i (\lc_x b)^{r-i} (\lc_x\sqfp b)^r \Bigl(\lc_x \prod\limits_{\ell=1}^L c_\ell\Bigr)^r
              \bigl(\deg_x a-\deg_x b\bigr)^i.
  \end{alignat*}
\item\label{lemma:8:2} If $\deg_x a\leq\deg_x b$, then
  \[
   \frac{D_x^i h}{h} = \frac{N_{r,i}}{c_0\bigl( b\sqfp b\prod_{\ell=1}^L c_\ell\bigr)^r}
  \]
  for some $N_{r,i}\in\K[x,y]$ with
  \begin{alignat*}1
    \deg_xN_{r,i}&=\deg_xc_0+r\Bigl(\deg_x\sqfp b+\deg_x b+\sum_{\ell=1}^L \deg_x c_\ell\Bigr) - i
       -\iverson{\omega\in\set N\land i>\omega}\delta,\\[-3pt]
    \deg_yN_{r,i}&\leq\deg_yc_0+ r\Bigl(\deg_y\sqfp b + \max\{\deg_y a,\deg_y b\} + \sum_{\ell=1}^L \deg_y c_\ell\Bigr),\\[-3pt]
    \lc_xN_{r,i}&=\left\{\begin{array}{ll}
(\lc_xc_0)(\lc_x b \sqfp b)^{r}\bigl(\lc_x\prod\limits_{\ell=1}^L c_\ell\bigr)^r
       \ff{\omega}i&\quad\text{if $\omega\not\in\set N$ or $i\leq\omega$;}\\
           (\lc_x N_{\omega+1})
           \bigl(\lc_x b\sqfp b\prod\limits_{\ell=1}^Lc_\ell\bigr)^{r-(\omega+1)}
           \ff{(-\delta-1)}{i-(\omega+1)}
       &\quad\text{if $\omega\in\set N$ and $i>\omega$,}
       \end{array}\right.
  \end{alignat*}
  where $\omega$, $\delta$, and $\lc_x N_{\omega+1}$ are as in Lemma~\ref{lemma:7}.(\ref{lemma:7:2}).
\end{enumerate}
\end{lemma}
\begin{proof}
 Both parts follow directly from the respective
 parts of Lemma~\ref{lemma:7} by multiplying numerator and denominator of the
 representations stated there by $\bigl( b\sqfp b\prod_{\ell=1}^L c_\ell\bigr)^{r-i}$.
\end{proof}

Since we will be frequently referring to the quantities in this lemma, it seems
convenient to adopt the following definition.

\begin{definition}\label{def:8}
 For a hyperexponential term~$h$, let
  \begin{alignat*}4
    \alpha &= \deg_x\sqfp b + \deg_x b + \sum_{\ell=1}^L\deg_x c_\ell,&\qquad
     \beta &= \deg_x a-\deg_x b - 1,\\[-3pt]
    \gamma &= \deg_y\sqfp b + \max\{\deg_ya,\deg_yb\} + \sum_{\ell=1}^L\deg_y c_\ell,&
    \omega &= \deg_xc_0+\sum_{\ell=1}^L e_\ell\deg_x c_\ell.
  \end{alignat*}
  If $\deg_xa\leq\deg_xb$ and $\omega\in\set N$, we further let $\delta$ be any integer with
  \[
   \omega+1\leq\delta\leq\deg_xc_0+(\omega+1)(\alpha-1)-
     \deg_x \biggl(c_0\Bigl( b\sqfp b\prod_{\ell=1}^Lc_\ell\Bigr)^{\omega+1}\frac{D_x^{\omega+1}h}h\biggr).
  \]
  Otherwise, if $\deg_xa>\deg_xb$ or $\omega\not\in\set N$, let $\delta=0$. Finally,
  we define the following flags:
  \begin{alignat*}1
    &\phi_1= \iverson{\tfrac{\lc_xa}{\lc_xb}\in\K},\qquad
    \phi_2= \iverson{\tfrac{\lc_xa}{\lc_xb}\in\K\land\beta=0},\\
    &\phi_3= \iverson{\,\tfrac ab\in\K(x)
      \land\forall\ \ell: (\deg_yc_\ell=0\lor e_\ell\in\set Z)
      \land
      \deg_yc_0\geq{\textstyle\sum\limits_{\ell=1}^L}e_\ell\deg_yc_\ell}\,.
  \end{alignat*}
\end{definition}

Note that none of these parameters depends on $r$ or~$i$.  The flags $\phi_k$
($k=1,2,3$) are in $\{0,1\}$, $\omega$ belongs to~$\K$, $\beta$ belongs to $\set
N\cup\{-1\}$, and all other parameters are positive integers.  The best
  value for $\delta$ is the right bound of the specified range, but since this
  value cannot be directly read of from the input, we do not insist that
  $\delta$ be equal to this value, but we allow $\delta$ to be any number
  between the bound from Lemma~\ref{lemma:delta} and the true degree drop. The
flags $\phi_1$ and $\phi_2$ will be used below in the ansatz for the telescoper,
$\phi_3$ will play a role afterwards in the ansatz for the certificate.

In terms of the parameters defined in Definition~\ref{def:8},
the degree bounds of Lemma~\ref{lemma:8} simplify to
\begin{alignat*}1
  &\deg_x N_{r,i}\leq\deg_xc_0+\alpha r+\max\{\beta,-1\} i-\iverson{\omega\in\set N\land i>\omega}\delta, \\
  &\deg_y N_{r,i}\leq\deg_yc_0+ \gamma r.
\end{alignat*}

\subsection{The Ansatz for the Telescoper}\label{sec:3.1}

Lemma~\ref{lemma:8} suggests reasonable choices for the degrees $d_i$ in
the ansatz for~$P$. In particular, our choice is based on the following
features of the formulas in Lemma~\ref{lemma:8}.

\begin{itemize}
\item The degree of the numerator in $(D_x^i h)/h$ varies with~$i$.
  A good choice for the degrees~$d_i$ will compensate for this variation,
  taking higher values for $d_i$ when the numerator of $(D_x^ih)/h$ has low degree,
  and vice versa. This is the key idea of the Verbaeten completion
  \citep{verbaeten74,verbaeten76,wegschaider97}.
\item The leading coefficients of $N_{r,i}$ ($i>0$) are polynomials in~$y$, but in case
  \ref{lemma:8:2}, most of them are $\K$-multiples of each
  other. When $a$ and $b$ are such that $(\lc_x a)/(\lc_x b)\in\K$, then this is
  also true in case~\ref{lemma:8:1}. We will use this fact for eliminating several
  equations at the cost of a single variable.
\end{itemize}

Before describing the ansatz for $P$ in full generality, we motivate the
construction by an example.

\begin{example}\label{ex:8}
  Suppose that $h$ is hyperexponential with $\lc_xa=\lc_xb$, $\beta=1$
  (case~\ref{lemma:8:1} of Lemma~\ref{lemma:8}), and $\deg_xc_0=0$.

  Let $r=5$ and $d=7$. We want to choose $d_i$ such that $\max_{i=0}^5 d_i=7$ and
  the ansatz
  \[
    P = \sum_{i=0}^5\sum_{j=0}^{d_i} p_{i,j}x^j D_x^i
  \]
  leads to ``many'' variables but only ``few'' equations. The choice with most
  variables is clearly to set $d_i=d=7$ for all~$i$. But this ansatz leads to
  quite many equations. Each term $x^j D_x^i$ contributes to the common
  numerator a polynomial $x^j N_{5,i}$ whose degree in~$x$ is $5\alpha+i+j$ and
  whose degree in~$y$ is at most~$5\gamma$. Because of the term $x^7 D_x^5$, we
  must expect up to $(5\alpha+13)(5\gamma+1)$ terms in the numerator.  This is
  the expected number of equations in the linear system resulting from
  coefficient comparison.

  If we remove the term $x^7 D_x^5$ from the ansatz, i.e., if we choose
  $d_0=\cdots=d_4=7$, $d_5=6$, then the number of equations drops to
  $(5\alpha+12)(5\gamma+1)$ because all terms $x^j D_x^i$ other than $x^7 D_x^5$
  contribute only polynomials $x^j N_{5,i}$ of lower degree.
  We save $5\gamma+1$ equations at the cost of removing a single variable.
  Removing also the terms $x^7 D_x^4$ and $x^6 D_x^5$ lowers the number of
  equations further to $(5\alpha+11)(5\gamma+1)$, and in general, for any $0\leq
  w\leq 5$, choosing $d_i=7-(w+i-5)^+$ ($i=0,\dots,5$) leads to
  $(5\alpha+13-w)(5\gamma+1)$ equations. The number of variables is
  $(5+1)(7+1)-\sum_{k=1}^w k = 48 - \tfrac12 w(w+1)$.

  If $w>1$, we can introduce $w-1$ new variables by exploiting the second
  feature of the formulas in Lemma~\ref{lemma:8} as follows. Consider
  the choice $w=3$, i.e., the terms $x^jD_x^i$ with $i+j\geq10$
  have been removed from the ansatz. We reintroduce the terms
  $x^7D_x^3$, $x^6D_x^4$, $x^5D_x^5$ by adding
  \[
    p_{3,7}\bigl((\deg_x a-\deg_x b)^2 x^7 D_x^3 - x^5 D_x^5\bigr) +
    p_{4,6}\bigl((\deg_x a-\deg_x b) x^6 D_x^4 - x^5 D_x^5\bigr)
  \]
  to the ansatz, getting back the two variables $p_{3,7}$ and $p_{4,6}$
  but no new equations, because,
  according to Lemma~\ref{lemma:8}.\eqref{lemma:8:1}, the assumption
  $\lc_xa=\lc_xb$ implies
  \[
    (\deg_x a-\deg_x b)^2\lc_x N_{5,3} = \lc_x N_{5,5}
    \quad\text{and}\quad
    (\deg_x a-\deg_x b)\lc_x N_{5,4} = \lc_x N_{5,5}.
  \]
  The final ansatz is depicted in Figure~\ref{fig:2}. A bullet at $(i,j)$
  represents a variable $p_{i,j}$ in the ansatz. White bullets correspond
  to the reintroduced variables $p_{3,7}$ and $p_{4,6}$ which do not
  affect the number of equations.

  \begin{figure}
    \centerline{%
    \begin{picture}(105,135)(-15,-15)
      \multiput(0,0)(15,0){6}{\circle*{5}}
      \multiput(0,15)(15,0){6}{\circle*{5}}
      \multiput(0,30)(15,0){6}{\circle*{5}}
      \multiput(0,45)(15,0){6}{\circle*{5}}
      \multiput(0,60)(15,0){6}{\circle*{5}}
      \multiput(0,75)(15,0){5}{\circle*{5}}
      \multiput(0,90)(15,0){4}{\circle*{5}}
      \multiput(0,105)(15,0){3}{\circle*{5}}
      \multiput(60,90)(-15,15){2}{\circle{5}}
      \put(-5,0){\line(-1,0){5}}\put(-14,-2){\llap{\footnotesize $0$}}
      \put(-5,105){\line(-1,0){5}}\put(-14,103){\llap{\footnotesize $d{=}7$}}
      \put(0,-5){\line(0,-1){5}}\put(0,-18){\clap{\footnotesize $0$}}
      \put(75,-5){\line(0,-1){5}}\put(75,-18){\clap{\footnotesize $r{=}5$}}
      \put(40,110){$\overbrace{\kern45pt}^{w=3}$}
    \end{picture}}

    \caption{The ansatz for $P$ discussed in Example~\ref{ex:8}}\label{fig:2}
  \end{figure}
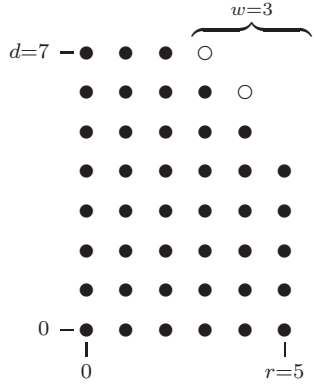

\end{example}

The general form of our ansatz for the telescoper is given in the following
lemma. The first case is like in the example above when $\beta>0$. For $\beta=0$,
no degree compensation is possible because all $N_{r,i}$ have the same degree.
But if $(\lc_x a)/(\lc_x b)\in\K$, it is still possible to save some equations
by exploiting the linear dependence among the leading terms. In the second
case, there is always a degree compensation possible, but unlike in the example
above, terms are removed for indices~$i$ close to zero rather than close to~$r$.
When $\omega\in\set N$, we provide an alternative ansatz which takes the degree drop
$\delta$ into account. Common to all cases are the two basic principles of
choosing $d_i$ such as to compensate for the different degrees of the $N_{r,i}$
in Lemma~\ref{lemma:8}, and of installing some additional variables by
exploiting the knowledge about the leading terms of the~$N_{r,i}$. For the size
of the cutoff, we use a new integer parameter~$w$, whose optimal value will be
determined later.

\begin{lemma}\label{lemma:9} Let $h$ be a hyperexponential term, $r\geq1$, $d\geq0$.
\begin{enumerate}
\item\label{lemma:9:1} Suppose that $\deg_xa>\deg_xb$.
  Let $0\leq w \leq \min\{r, d/\beta\}$ ($w:=0$ if $\beta=0$),
  $d_i:=d - \beta(w+i-r)^+ -\phi_2$ ($i=0,\dots,r$), and
  \begin{alignat*}1
    P &= \sum_{i=0}^r \sum_{j=0}^{d_i} p_{i,j} x^j D_x^i\\
      &\quad{} + \iverson{\beta\neq0}\phi_1
        \sum_{i=r-w+1}^{r-1} p_{i,d_i+1}\Bigl(\bigl(\tfrac{\lc_x a}{\lc_x b}(\beta+1)\bigr)^i x^{d_i+1}D_x^i - x^{d_r+1} D_x^r\Bigr)\\
      &\quad{} + \phi_2 \sum_{i=0}^{r-1} p_{i,d_i+1}\Bigl(\bigl(\tfrac{\lc_x a}{\lc_x b}\bigr)^i x^{d_i+1}D_x^i - x^{d_r+1} D_x^r\Bigr).
  \end{alignat*}
  Let $N=c_0\bigl(b^\ast b\prod_{\ell=1}^Lc_\ell\bigr)^r(Ph)/h$.
  Then
  \[
    \deg_x N\leq \deg_xc_0+ d + (\alpha+\beta)r-\beta w -\phi_2
    \quad\text{and}\quad
    \deg_y N\leq \deg_yc_0+ \gamma r.
  \]
\item\label{lemma:9:2} Suppose that $\deg_xa\leq\deg_xb$.
  Let $0\leq w\leq\min\{d+1,r+1\}$. Let $d_i:=d-(w-i)^+$ ($i=0,\dots,r$), and
  \begin{alignat*}1
    P &= \sum_{i=0}^r \sum_{j=0}^{d_i} p_{i,j}x^j D_x^i
    +  \sum_{i=1}^{w-1} p_{i,d_i+1}\Bigl(x^{d_i+1}D_x^i - \ff{\omega}{i}x^{d_0+1}\Bigr).
  \end{alignat*}
  Let $N=c_0\bigl(b^\ast b\prod_{\ell=1}^Lc_\ell\bigr)^r(Ph)/h$.
  Then
  \begin{alignat*}1
    \deg_x N\leq \deg_xc_0+d+\alpha r-w\quad\text{and}\quad
    \deg_y N\leq \deg_yc_0+\gamma r.
  \end{alignat*}
\item[(2$'$)]\label{lemma:9:2b} Suppose that $\deg_xa\leq\deg_xb$ and $\omega\in\set N$.
  Let $\omega\leq w\leq\min\{d-\delta+1,r+1\}$.
  Let $d_i:=d-(w-i)^+-\iverson{i\leq\omega}\delta$ ($i=0,\dots,r$), and
  \begin{alignat*}1
    P &= \sum_{i=0}^r \sum_{j=0}^{d_i} p_{i,j}x^j D_x^i
      + \sum_{i=1}^{\omega} p_{i,d_i+1}\Bigl(x^{d_i+1}D_x^i - \ff{\omega}{i}x^{d_0+1}\Bigr)\\
      &\qquad{} + \sum_{i=\omega+2}^{w-1} p_{i,d_i+1}
           \Bigl(x^{d_i+1}D_x^i - \ff{(-\delta-1)}{i-(\omega+1)}x^{d_{\omega+1}+1}D_x^{\omega+1}\Bigr).
  \end{alignat*}
  (See Figure~\ref{fig:3} for an illustration of the shape of $P$ in this case.)

  Let $N=c_0\bigl(b^\ast b\prod_{\ell=1}^Lc_\ell\bigr)^r(Ph)/h$.
  Then
  \[
    \deg_x N\leq \deg_xc_0+d+\alpha r-w-\delta
    \quad\text{and}\quad
    \deg_y N\leq \deg_yc_0+\gamma r.
  \]
\end{enumerate}
\end{lemma}
\begin{proof}
\begin{enumerate}
\item We apply Lemma~\ref{lemma:8}.(\ref{lemma:8:1}) to each term in the ansatz for~$P$.
  The claim about $\deg_y N$ follows directly from the bound on $\deg_y N_{r,i}$ there.
  For the bound on $\deg_x N$, first observe that
  \begin{alignat*}1
    \deg_x x^j N_{r,i}&\leq d_i + \deg_x c_0+\alpha r + \beta i\\
      &= \deg_xc_0 + d +\alpha r + \beta i- \beta(w+i-r)^+ -\phi_2 \\
      &= \deg_xc_0 + d + \alpha r + \beta( i - \max\{w+i-r,0\} ) -\phi_2 \\
      &\leq \deg_xc_0+d+\alpha r+\beta(r-w) -\phi_2
  \end{alignat*}
  for all $i,j$ with $0\leq i\leq r$ and $0\leq j\leq d_i$. This settles the terms
  coming from the double sum. For the terms in the first single sum, which only appears
  when $\beta\neq0$, we have
  \begin{alignat*}1
     &\deg_x x^{d_i+1}N_{r,i}
     =\deg_x x^{d_r+1}N_{r,r}\\
     ={}&\deg_xc_0+d+\alpha r+\beta(r-w) + 1 -\phi_2
  \end{alignat*}
  and
  \[
   \bigl(\tfrac{\lc_x a}{\lc_x b}(\beta+1)\bigr)^i \lc_x N_{r,i} = \lc_x N_{r,r}
  \]
  for $i=r-w+1,\dots,r-1$. This implies
  \begin{alignat*}1
    &\deg_x\Bigl(\bigl(\tfrac{\lc_x a}{\lc_x b}(\beta+1)\bigr)^i x^{d_i+1}N_{r,i}
        -  x^{d_r+1}N_{r,r}\Bigr)\leq\deg_xc_0+d+\alpha r+\beta(r-w)-\phi_2,
  \end{alignat*}
  as desired. The argument for the second single sum, which only appears when $\beta=0$, is
  analogous.
\item Now we use Lemma~\ref{lemma:8}.(\ref{lemma:8:2}). Again, the claim about
  $\deg_y N$ follows immediately. For the bound on $\deg_xN$, first observe that
  \begin{alignat*}1
    \deg_x x^j N_{r,i}&\leq d_i+\deg_xc_0+\alpha r - i\\
       &=\deg_xc_0 + d + \alpha r - i - (w-i)^+\\
       &\leq\deg_xc_0 + d + \alpha r - w.
  \end{alignat*}
  This settles the terms in the double sum. For the terms in the single sum, we have
  \[
    \deg_x x^{d_i+1} N_{r,i} = \deg_x x^{d_0+1} N_{r,0} = \deg_xc_0+d+\alpha r - w + 1
  \]
  and $\lc_x N_{r,i} = \lc_x \ff\omega i N_{r,0}$ for $i=1,\dots,w-1$, and therefore
  \[
    \deg_x\Bigl(x^{d_i+1}N_{r,i} - \ff\omega i x^{d_0+1}N_{r,0}\Bigr)\leq
      \deg_xc_0+d+\alpha r - w.
  \]
\item[(2$'$)] In this case, the terms in the double sum contribute polynomials of degree
  \begin{alignat*}1
    \deg_x x^j N_{r,i}&\leq d_i+\deg_xc_0+\alpha r- i - \iverson{i>\omega}\delta\\
    &=\deg_xc_0 + d + \alpha r - i - (w-i)^+ - \iverson{i\leq\omega}\delta - \iverson{i>\omega}\delta\\
    &\leq\deg_xc_0 + d + \alpha r - w - \delta.
  \end{alignat*}
  For the terms in the first single sum, we have
  \[
    \deg_x x^{d_i+1} N_{r,i} = \deg_x x^{d_0+1} N_{r,0} = \deg_xc_0+d+\alpha r - w + 1 - \delta
  \]
  and $\lc_x N_{r,i} = \lc_x \ff\omega i N_{r,0}$ for $i=1,\dots,\omega$, and therefore
  \[
    \deg_x\Bigl(x^{d_i+1}N_{r,i} - \ff\omega i x^{d_0+1}N_{r,0}\Bigr)\leq
      \deg_xc_0+d+\alpha r - w - \delta.
  \]
  Similarly, for the terms in the second single sum, we have
  \[
    \deg_x x^{d_i+1} N_{r,i} = \deg_x x^{d_{\omega+1}+1} N_{r,\omega+1} \leq \deg_xc_0+d+\alpha r - w + 1 - \delta.
  \]
  If the inequality is strict, we are done. Otherwise, $\delta$ is maximal and we have
  $\lc_x N_{r,i}= \lc_x\ff{(-\delta-1)}{i-(\omega+1)} N_{r,\omega+1}$ for $i=\omega+2,\dots,w-1$,
  and therefore
  \[
    \deg_x\Bigl(x^{d_i+1}N_{r,i} - \ff{(-\delta-1)}{i-(\omega+1)} x^{d_{\omega+1}+1}N_{r,\omega+1}\Bigr)\leq
      \deg_xc_0+d+\alpha r - w - \delta,
  \]
  and we are also done.
\end{enumerate}
\leavevmode
\end{proof}

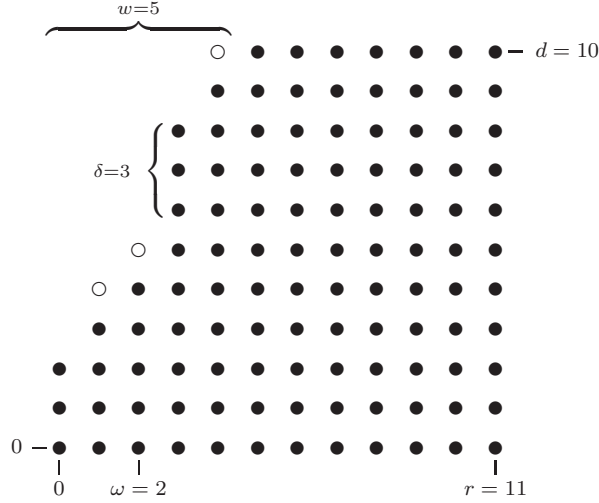
\begin{figure}
  \centerline{%
    \begin{picture}(195,180)(-15,-15)
      \multiput(0,0)(15,0){12}{\circle*{5}}
      \multiput(0,15)(15,0){12}{\circle*{5}}
      \multiput(0,30)(15,0){12}{\circle*{5}}
      \multiput(15,45)(15,0){11}{\circle*{5}}\put(15,60){\circle{5}}
      \multiput(30,60)(15,0){10}{\circle*{5}}\put(30,75){\circle{5}}
      \multiput(45,75)(15,0){9}{\circle*{5}}
      \multiput(45,90)(15,0){9}{\circle*{5}}
      \multiput(45,105)(15,0){9}{\circle*{5}}
      \multiput(45,120)(15,0){9}{\circle*{5}}
      \multiput(60,135)(15,0){8}{\circle*{5}}\put(60,150){\circle{5}}
      \multiput(75,150)(15,0){7}{\circle*{5}}
      \put(-5,0){\line(-1,0){5}}\put(-14,-2){\llap{\footnotesize $0$}}
      \put(175,150){\line(-1,0){5}}\put(180,148){\rlap{\footnotesize $d=10$}}
      \put(0,-5){\line(0,-1){5}}\put(0,-18){\clap{\footnotesize $0$}}
      \put(30,-5){\line(0,-1){5}}\put(30,-18){\clap{\footnotesize $\omega=2$}}
      \put(165,-5){\line(0,-1){5}}\put(165,-18){\clap{\footnotesize $r=11$}}
      \put(42,102.5){\llap{${\scriptstyle\delta=3}\ \left\{\rule{0pt}{20pt}\right.$}}
      \put(-5,155){$\overbrace{\kern70pt}^{w=5}$}
    \end{picture}}
  \caption{The ansatz for $P$ in case 2$'$ of Lemma~\ref{lemma:9}}\label{fig:3}
\end{figure}

Lemma~\ref{lemma:9} makes a statement on the number of equations to be expected
when the ansatz for $P$ is made in the form as indicated. This number of
equations is equal to the number of terms $x^iy^j$ in~$N$, and this number is
bounded by $(\deg_x N+1)(\deg_y N+1)$, for which upper bounds are stated in the
lemma. We also need to count the number of variables~$p_{i,j}$. This number is
easily obtained from the sum expressions given for $P$ in the various cases by
replacing all the summand expressions by~$1$. After some straightforward and elementary
simplifications which we do not want to reproduce here, the statistics
are as follows.

\begin{itemize}
\item In case 1, the number of variables is
  \begin{alignat*}1
    (r+1)(d+1) - \tfrac12\beta w(w+1) +
    \phi_1(w-1)^+-\phi_2.
  \end{alignat*}
\item In case 2, the number of variables is
  \[
    (r+1)(d+1) - \tfrac12 w(w+1) + (w-1)^+.
  \]
\item In case~2$'$, the number of variables is
  \[
    (r+1)(d+1) - \tfrac12 w(w+1) - \delta(\omega+1)+
    \omega + (w-\omega-2)^+.
  \]
\end{itemize}

These are only the variables coming from the degrees of freedom in the
telescoper~$P$.  We will next discuss the ansatz for the certificate~$Q$, which
will bring many additional variables, but, by a careful construction, no
additional equations.

\subsection{The Ansatz for the Certificate}\label{sec:3.2}

The design of the ansatz for the certificate is much simpler. Here, the goal is
to set up $Q$ in such a way that $(D_yQ)/h$ has the same
denominator and the same numerator degrees in $x$ and $y$ as $(Ph)/h$ does (in
order to not create more equations than necessary), and that $(D_yQ)/h$ cannot
become zero (in order to enforce that $P\neq0$ in every solution we find).

A direct calculation like in the proof of Lemma~\ref{lemma:7} confirms that the
first requirement is satisfied by choosing
\[
  Q = \frac{\sum\limits_{i=0}^{s_1}\sum\limits_{j=0}^{s_2} q_{i,j}x^i y^j}
           {c_0\bigl( b\sqfp b\prod_{\ell=1}^Lc_\ell\bigr)^{r-1}}h
\]
with
\begin{alignat*}1
  s_1&=\left\{\begin{array}{ll}
      \deg_xc_0+d+(\alpha+\beta)(r-1)-\beta w - \phi_2 - 1
      &\quad \text{in case 1 of Lemma~\ref{lemma:9};}\\
      \deg_xc_0+d+\alpha(r-1)- w &\quad \text{in case 2 of Lemma~\ref{lemma:9};}\\
      \deg_xc_0+d+\alpha(r-1)- w - \delta
       &\quad \text{in case 2$'$ of Lemma~\ref{lemma:9}}
    \end{array}\right.\\
\intertext{and}
  s_2&=\rlap{$\deg_yc_0+\gamma (r-1) + 1$}
      \hphantom{\left\{\rule{0pt}{2em}\right.}
      \begin{array}{ll}
        \hphantom{\deg_xc_0+d+(\alpha+\beta)(r-1)-\beta w - \phi_2 - 1}
      &\quad\kern-1.3pt\text{in all cases.}
      \end{array}
\end{alignat*}
This ansatz provides $(s_1+1)(s_2+1)$ variables. To ensure that $D_yQ\neq0$ for every
choice of~$q_{i,j}$, observe that $D_yQ=0$ can only happen if $h$ is a rational
function with respect to~$y$, meaning $a,b\in\K[x]$ and $c_\ell\in\K[x]$ for all $\ell$
with $e_\ell\not\in\set Z$. In this case, we have $D_yQ=0$ if and only if the $q_{i,j}$
are instantiated in such a way that the resulting $Q$ is free of~$y$, and this can
only happen if the choice of $q_{i,j}$ is made in such a way that the numerator degree
in $y$ is equal to the denominator degree in~$y$.
The denominator degree is
\[
  \sum_{\ell=1}^L (r-1-e_\ell)\deg_y c_\ell = \gamma(r-1) - \eta,
  \qquad\text{where }\eta=\sum_{\ell=1}^L e_\ell\deg_yc_\ell,
\]
which is less than $s_2=\deg_yc_0+\gamma(r-1)+1$ if and only if $\deg_yc_0+\eta+1>0$.
If we remove all the terms $q_{i,j}x^i y^j$ with $j=\gamma(r-1)-\eta$ from the
ansatz, no instantiation of the remaining $q_{i,j}$ can turn $Q$ into a term independent
of~$y$, so we can be sure that $D_yQ\neq0$ in this modified setup.
The number of variables in this modified ansatz is $(s_1+1)s_2$.
The flag $\phi_3$ defined in Definition~\ref{def:8} is set up in such a way
that we can in all cases assume an ansatz for $Q$ with $(s_1+1)(s_2+1-\phi_3)$ variables.
The following lemma summarizes the two versions of the ansatz for~$Q$.

\begin{lemma} Let $h$ be a hyperexponential term.
  \begin{enumerate}
  \item If $\max\{\deg_ya,\deg_yb\}>0$ or $\deg_yc_\ell>0$ for some $\ell$ with
    $e_\ell\not\in\set Z$, then for every $s_1,s_2\in\set N$ and every choice of
    $q_{i,j}\in\K$ where not all $q_{i,j}$ are equal to zero we have
    \[
      D_y\left(\frac{\sum\limits_{i=0}^{s_1}\sum\limits_{j=0}^{s_2} q_{i,j}x^iy^j}
           {c_0\bigl( b\sqfp b\prod_{\ell=1}^Lc_\ell\bigr)^{r-1}}h\right)\neq0.
    \]
  \item If $\deg_ya=\deg_yb=0$ and $\deg_yc_\ell=0$ for all $\ell$ with
    $e_\ell\not\in\set Z$, then for every $s_1,s_2\in\set N$ and every choice of
    $q_{i,j}\in\K$ where not all $q_{i,j}$ are equal to zero we have
    \[
      D_y\left(\frac{\sum\limits_{i=0}^{s_1}
        \biggl(\,\sum\limits_{j=0}^{(r-1)\gamma-\eta-1}\!
          q_{i,j}x^iy^j
          + \!\sum\limits_{j=(r-1)\gamma-\eta+1}^{s_2}\! q_{i,j}x^iy^j\biggr)}
           {c_0\bigl( b\sqfp b\prod_{\ell=1}^Lc_\ell\bigr)^{r-1}}h\right)\neq0,
    \]
    where $\eta=\sum_{\ell=1}^L e_\ell\deg_yc_\ell$.
  \end{enumerate}
\end{lemma}

\section{Solving the Inequalities}\label{sec:5}

As the result of the previous section, we obtain counts for the number of
variables and the number of equations for a particular family of ansatzes which
are parameterized by the desired order~$r$ and degree~$d$ of the telescoper,
various Greek parameters introduced in Definition~\ref{def:8}, which measure the
input, and one additional parameter~$w$ by which the shape of the ansatz can be
modulated. A sufficient condition for the existence of a solution of order (at
most)~$r$ and degree (at most)~$d$ is
\[
  \mathrm{\#vars}(r,d,w) - \mathrm{\#eqns}(r,d,w) > 0.
\]
For any particular choice of $w$ from the ranges specified for the various cases
in Lemma~\ref{lemma:9}, we obtain a valid sufficient condition connecting $r$
and~$d$ via the Greek parameters.
Any of these conditions defines a region in
$\set N^2$ which is inside the gray region from the introduction. To make this
region as large as possible (and hence, as equal as possible to the gray
region), we will choose $w$ in such a way that the left hand side, considered as
a function in~$w$, is maximal.

It comes in handy that $\mathrm{\#vars}(r,d,w) - \mathrm{\#eqns}(r,d,w)$ is a
(piecewise) quadratic polynomial with respect to~$w$, so the optimal choice of
$w$ is easily found by equating its derivative with respect to $w$ to zero and
rounding the solution to the nearest integer. If this point is outside the range
to which $w$ is constrained, then the maximum is assumed at one of the two
boundary points of the range.

The following theorem, which is the main result of this article, contains the
bounds which we obtained by applying this reasoning to the explicit expressions
derived for $\mathrm{\#vars}(r,d,w)$ and $\mathrm{\#eqns}(r,d,w)$ in the
previous section for the various cases to be considered.

\begin{theorem}\label{thm:1}
  Let
  \[
    h= c_0\exp\Bigl(\frac{a}{b}\Bigr)\prod_{\ell=1}^L c_\ell^{e_\ell}
  \]
  be a hyperexponential term and let $\alpha,\beta,\gamma,\delta,\omega,\phi_1,\phi_2,\phi_3$
  be as in Definition~\ref{def:8} and set $\psi=\gamma+\phi_3-2$.
  Then a creative telescoping relation for $h$ of order~$r$ and degree~$d$ exists whenever
  \[
     r\geq \psi + 1\qquad\text{and}\qquad
     d > \frac{\vartheta\,r + \varphi}{r-\psi} \ ,
  \]
  where $\vartheta$ and $\varphi$ are defined as follows.
  \begin{enumerate}
  \item\label{thm:1:1} If $\deg_x a>\deg_x b$, let
    \begin{alignat*}1
      \vartheta &=(\alpha+\beta)(2\gamma-1+\phi_3)+\gamma-1,\\
      \varphi &= \deg_xc_0+(\alpha+\beta+1)\deg_yc_0+(\gamma-2+\phi_3)(\deg_xc_0-\alpha-\beta-\phi_2)\\
      &\qquad{}
      - (1-\phi_2)(\gamma-2+\phi_3)^+\bigl(\phi_1+\tfrac12\beta(\gamma-1+\phi_3)\bigr).
    \end{alignat*}
  \item\label{thm:1:2} If $\deg_x a\leq\deg_x b$, let
    \begin{alignat*}1
      \vartheta &= \alpha(2\gamma-1+\phi_3)-1,\\
      \varphi &=\deg_xc_0+\alpha\deg_yc_0+(\gamma-2+\phi_3)(\deg_xc_0+1-\alpha)\\
         &\quad{}-\tfrac12(\gamma-2+\phi_3)^+(\gamma+1+\phi_3).
    \intertext{If furthermore $\omega\in\set N$ and $\gamma-1+\phi_3>\omega$ and $\delta=\omega+1$, then $\varphi$ can be
    replaced by}
       \varphi'&=\varphi-\delta(\gamma-2+\phi_3-\omega)+1.
    \end{alignat*}
  \end{enumerate}
\end{theorem}

\begin{proof}
  \begin{enumerate}
  \item Suppose $\deg_x a>\deg_x b$.
    According to the calculations done in the previous section, in this case there exists
    an ansatz with
    \[
      (r+1)(d+1) - \tfrac12\beta w(w+1) + \phi_1 (w-1)^+-\phi_2
    \]
    variables coming from the telescoper~$P$,
    \[
     \bigl(\deg_x c_0 + d + (\alpha+\beta)(r-1) - \beta w -
     \phi_2\bigr)\bigl(\deg_yc_0+\gamma(r-1)+2-\phi_3\bigr)
    \]
    variables coming from the certificate~$Q$, and
    \[
     \bigl(\deg_xc_0+d+(\alpha+\beta)r-\beta w - \phi_2 + 1\bigr)
     \bigl(\deg_yc_0+\gamma r+1\bigr)
    \]
    equations. Therefore, a creative telescoping relation exists provided that
      \begin{alignat*}1
        &(r+1)(d+1)-\tfrac12\beta w(w+1)+\phi_1(w-1)^+-\phi_2\notag\\
        &\quad{}+(\deg_xc_0+d+(\alpha+\beta)(r-1)-\beta w-\phi_2)(\deg_yc_0+\gamma(r-1)+2-\phi_3)\\
        &\qquad{}-(\deg_xc_0+d+(\alpha+\beta)r-\beta w-\phi_2+1)(\deg_yc_0+\gamma r+1)>0.\notag
      \end{alignat*}
    For $r\geq\gamma-1+\phi_3$, this inequality is equivalent to
    \begin{alignat}1
      d>&\Bigl(\!\bigl((\alpha+\beta)(2\gamma-1+\phi_3\bigr)+\gamma-1\bigr)r + \deg_xc_0+(\alpha+\beta+1)\deg_yc_0\notag\\
       &\quad{}+(\gamma-2+\phi_3)(\deg_xc_0-\alpha-\beta-\phi_2)\label{eq:3}\\
       &\qquad{}
         + \tfrac12\beta w(w-2\gamma+3-2\phi_3)- \phi_1(w-1)^+\Bigr)
         \Big/\Bigl(r-\gamma+2-\phi_3\Bigr).\notag
    \end{alignat}
    The choice $w=0$ proves the claim when $\phi_2=1$ or $\gamma\leq 1-\phi_3$.
    Now suppose that $\phi_2=0$ and $\gamma>1-\phi_3$.
    The claimed estimate is obtained for the choice $w=\gamma-1+\phi_3>0$.
    We have to show that this choice is admissible, i.e., that $1\leq\gamma-1+\phi_3\leq\min\{r,d/\beta\}$.
    Because of $\gamma>1-\phi_3$, the lower bound is clear,
    and $r\geq\gamma-1+\phi_3$ holds by assumption.
    To see that $\gamma-1+\phi_3\leq d/\beta$, observe that the right hand side of \eqref{eq:3}
    converges to $(\alpha+\beta)(2\gamma-1+\phi_3)+\gamma-1$ for $r\to\infty$.
    Since its numerator is nonnegative (as checked by a straightforward calculation), it
    follows that this inequality implies
    \[
      d>(\alpha+\beta)(2\gamma-1+\phi_3)+\gamma-1\geq\beta(\gamma-1+\phi_3),
    \]
    as desired.
  \item Now assume $\deg_x a\leq\deg_x b$.
    From the counts of variables and equations in the ansatz described in
    Lemma~\ref{lemma:9}.(\ref{lemma:9:2}), we find that a creative telescoping equation
    exists provided that
    \begin{alignat*}1
      &(r+1)(d+1)-\tfrac12 w(w+1)+(w-1)^+\notag\\
      &\quad{}+(\deg_xc_0+d+\alpha(r-1)-w+1)(\deg_yc_0+\gamma(r-1)+2-\phi_3)\\
      &\qquad{}-(\deg_xc_0+d+\alpha r-w+1)(\deg_yc_0+\gamma r+1)>0.\notag
    \end{alignat*}
    For $r\geq\gamma-1+\phi_3$, this inequality is equivalent to
    \begin{alignat*}1
      d>\Bigl(&(\alpha(2\gamma-1+\phi_3)-1)r +
          \deg_xc_0+\alpha \deg_yc_0 + (\gamma-2+\phi_3)(\deg_xc_0+1-\alpha)
          \\
      &\quad{}
         + (\tfrac32-\gamma-\phi_3)w + \tfrac12 w^2 - (w-1)^+
      \Bigr)\Big/\Bigl(r-\gamma+2-\phi_3\Bigr).
    \end{alignat*}
    Regardless of the choice of~$w$, the right hand side is at least
    $\alpha(2\gamma-1+\phi_3)-1$.
    Similar as before, the claimed bound follows on one hand from the choice
    $w=0$ and on the other hand, if $\gamma>1-\phi_3$, from the choice $w=\gamma-1+\phi_3$,
    which also in this case is in the required range because $1\leq \gamma-1+\phi_3\leq
    \alpha(2\gamma-1+\phi_3)-1<d$
    and $\gamma-1+\phi_3\leq r$.

    The second estimate is obtained from the alternative ansatz from Lemma~\ref{lemma:9}.(2$'$).
    The inequality in this case is
    \begin{alignat*}1
      &(r+1)(d+1)-\tfrac12 w(w+1)-\delta(\omega+1)+\omega+(w-\omega-2)^+\notag\\
      &\quad{}+(\deg_xc_0+d+\alpha(r-1)-w-\delta+1)(\deg_yc_0+\gamma(r-1)+2-\phi_3)\\
      &\qquad{}-(\deg_xc_0+d+\alpha r-w-\delta+1)(\deg_yc_0+\gamma r+1)>0,
    \end{alignat*}
    which for $r\geq\gamma-1+\phi_3$ and $w=\gamma-1+\phi_3$ is equivalent to
    \[
     d>\frac{(\alpha(2\gamma-1+\phi_3)-1)r+\varphi'}{r-\gamma+2-\phi_3}.
    \]
    It remains to show that the choice $w=\gamma-1+\phi_3$ is compatible
    with the range restrictions for $w$ applicable in the present case.
    While the requirements $\omega\leq\gamma-1+\phi_3\leq r+1$
    are satisfied by assumption, the requirement $\gamma-1+\phi_3\leq d-\delta+1$
    is less obvious. A sufficient condition is
    \[
      \frac{(\alpha(2\gamma-1+\phi_3)-1)r+\varphi'}
           {r-\gamma+2-\phi_3}\geq\gamma-2+\phi_3+\delta.
    \]
    It can be shown easily with Collins's cylindrical algebraic decomposition
    algorithm \citep{collins75,caviness98} (e.g., with its implementation in
    Mathematica~\citep{strzebonski00,strzebonski06}) that this latter inequality
    follows from $\deg_xc_0\geq0$, $\deg_yc_0\geq0$, $\alpha\geq1$,
    $r\geq\gamma-1+\phi_3\geq\omega+1\geq1$, $\delta=\omega+1$, $\phi_3(\phi_3-1)=0$,
    and
    \begin{alignat*}1
     &\varphi'=\deg_xc_0+\alpha\deg_yc_0+\delta\omega+1+(\gamma-2+\phi_3)(\deg_xc_0-\alpha-\tfrac12(\gamma-1+\phi_3)-\delta).
    \end{alignat*}
    This completes the proof.
  \end{enumerate}
\end{proof}

As we do not claim that our bounds are sharp, no justification for the various
choices of $w$ are required in the proof. But of course, the choices were made
following the reasoning outlined before the theorem. For example, in case~1 the
main inequality is
\begin{alignat*}1
  &(r+1)(d+1+\phi_2)-\tfrac12\beta w(w+1)+\phi_1(w-1)^+-\phi_2\notag\\
  &\quad{}+(\deg_xc_0+d+(\alpha+\beta)(r-1)-\beta w-\phi_2+1)(\deg_yc_0+\gamma(r-1)+2-\phi_3)\\
  &\qquad{}-(\deg_xc_0+d+(\alpha+\beta)r-\beta w-\phi_2+1)(\deg_yc_0+\gamma r+1)>0.
\end{alignat*}
Differentiating the left hand side with respect to $w$ gives
\[
 -\beta w - \tfrac32\beta + \beta\gamma + \phi_1 + \beta\phi_3,
\]
which vanishes for $w=\gamma-\tfrac32+\phi_3+\phi_1/\beta$. The unique nearest integer
point is $\lfloor\gamma-\tfrac32+\phi_3+\phi_1/\beta\rceil=\gamma-1+\phi_3$ when
$\phi_1/\beta\neq1$. When $\phi_1/\beta=1$, there are two nearest integer
points $\gamma-1+\phi_3$ and $\gamma+\phi_3$, and since the maximum is exactly between them
and quadratic parabolas are symmetric about their extremal points, the values at
$\gamma-1+\phi_3$ and $\gamma+\phi_3$ agree. In conclusion, the choice $w=\gamma-1+\phi_3$ is optimal
in both cases.

The calculations for the other cases are similar. But note that having chosen
$w$ optimally does not imply that the bounds given in the Theorem~\ref{thm:1}
are tight, because the whole argument relies on counting variables and equations
for the particular ansatz family introduced in Section~\ref{sec:ansatz}, and we
cannot claim that this shape is best possible. Recall that we aim at an ansatz
for which the number of solutions of the resulting linear system is equal to (or
at least not much larger than) the difference between number of variables and
number of equations. One way of measuring the quality of our ansatz, and hence
the tightness of our bounds, is to compare the region of all points $(r,d)$
where an ansatz for order~$r$ and degree~$d$ actually has a solution (the ``gray
region'' from the introduction) with the region of all points $(r,d)$ for which
Theorem~\ref{thm:1} guarantees the existence of a solution. The following
collection of examples shows that there are cases where Theorem~\ref{thm:1} is
extremely accurate as well as cases where there is a clear gap between the
predicted shape and the actual shape of the gray region. As a reference ansatz
for experimentally determining in the examples whether a specific point $(r,d)$
belongs to the gray region, we checked whether the naive ansatz where
$d_0=d_1=\dots=d_r$ (i.e., $w=0$) as a solution, because every solution of some
refined ansatz with $w>0$ is also a solution of the ansatz with $w=0$.  It is
not guaranteed however that this ansatz covers all creative telescoping
relations. Additional relations at points $(r,d)$ outside of what we indicate as
the gray region may exist. For example, when our ansatz leads to a solution
$(P,Q)$ in which all the polynomial coefficients of $P$ share a nontrivial
common factor $f\in\K[x]$, then $(P/f,Q/f)$ is another relation with a
telescoper of lower degree. This phenomenon can often be observed for the
minimal order telescoper, but as we do not know of any efficient way of
detecting it also for the nonminimal ones, we can unfortunately not take it into
account in the figures.

\begin{example}\label{ex:13}
\begin{enumerate}
\item\label{ex:13:1} Consider the term $h=u\exp(v)$ where
  \begin{alignat*}1
    u &= 7 x^3 y^3+8 x^3 y^2+9 x^3 y+3 x^3+10 x^2 y^3+2 x^2 y^2+3 x^2 y+9 x^2\\
    &\quad{}+7 x y^3+4 x y^2+5 x y+3 x+9 y^3+6 y^2+6 y+1,\\
    v &= 6 x^3 y^3+4 x^3 y^2+x^3 y+9 x^3+8 x^2 y^3+8 x^2 y^2+2 x^2 y+8 x^2\\
    &\quad{}+3 x y^3+7 x y^2+4 x y+8 x+5 y^3+2 y^2+7 y+6.
  \end{alignat*}
  We are in case \ref{thm:1:1} of Theorem~\ref{thm:1} and have
  $\alpha=0$, $\beta=2$, $\gamma=3$, $\phi_1=\phi_2=\phi_3=0$,
  $\deg_xc_0=\deg_yc_0=3$.
  According to the theorem, we expect creative telescoping relations
  for all $(r,d)$ with $r\geq2$ and $d>(12r+11)/(r-1)$.
  Figure~\ref{fig:4}.(a) depicts the curve $(12r+11)/(r-1)$ together with
  the gray region. In this example, the gray region consists exactly of the
  integer points above the curve: the bound is as tight as can be.
\item\label{ex:13:2} Now consider the term $h=\exp(u)/v$ where
  \begin{alignat*}1
    u &= 4 x^2 y^2+7 x^2 y+9 x^2+5 x y^2+2 x y+3 x+5 y^2+y+6,\\
    v &= 6 x^2 y^2+10 x^2 y+6 x^2+9 x y^2+5 x y+8 x+8 y^2+10 y+8.
  \end{alignat*}
  We are again in case \ref{thm:1:1} of the theorem and we have
  $\alpha=2$, $\beta=1$, $\gamma=4$, $\phi_1=\phi_2=\phi_3=0$,
  $\deg_xc_0=\deg_yc_0=2$.
  The estimate from Theorem~\ref{thm:1} is now $d>(24r-9)/(r-2)$, which is depicted together
  with the gray region in Figure~\ref{fig:4}.(b). In this case, the bound is not
  sharp.
\item\label{ex:13:3} Now let $h$ be the rational function from the introduction.
  Then we are in case \ref{thm:1:2} of the theorem and we have
  $\alpha=3$, $\beta=-1$, $\gamma=3$, $\omega=-1$, $\delta=0$, $\phi_1=1$, $\phi_2=0$, $\phi_3=1$,
  $\deg_xc_0=\deg_yc_0=2$.
  The bound from the theorem is now $d>(17r+3)/(r-2)$, which is shown together with
  the gray region in Figure~\ref{fig:4}.(c). The curve correctly predicts all
  the degrees except for the minimal order recurrence, where the true degree is
  one less than predicted.
\item\label{ex:13:4} Next, let $h=u/v$ with
  \begin{alignat*}1
    u &= 4 x^2 y^2+7 x^2 y+9 x^2+5 x y^2+2 x y+3 x+5 y^2+y+6,\\
    v &= \bigl(6 x^2 y^2+10 x^2 y+6 x^2+9 x y^2+5 x y+8 x+8 y^2+10 y+8\bigr)\\
      &\quad\times\bigl(8 x^2 y^2+7 x^2 y+4 x^2+5 x y^2+3 x y+7 x+9 y^2+7 y+7\bigr).
  \end{alignat*}
  This term is also covered by case \ref{thm:1:2} of the theorem, and we have
  $\alpha=4$, $\beta=-1$, $\gamma=4$, $\omega=-2$, $\delta=-1$, $\phi_1=1$, $\phi_2=0$, $\phi_3=0$,
  $\deg_xc_0=\deg_yc_0=2$. The estimate $d>(27r+3)/(r-2)$ from the theorem is
  correct but not tight, as shown in Figure~\ref{fig:4}.(d).
\item\label{ex:13:5} Finally, let $h=\sqrt{u}$ with
  \begin{alignat*}1
    u&= 4 x^2 y^6+8 x^2 y^5+2 x^2 y^4+7 x^2 y^3+7 x^2 y^2+2 x^2 y+7 x^2+10 x y^6+7 x y^5+9 x y^4\\
     &\qquad{}+4 x y^3+5 xy^2+5 x y+7 x+4 y^6+3 y^5+2 y^4+8 y^3+3 y^2+7 y+2.
  \end{alignat*}
  Now the alternative bound of case \ref{thm:1:2} with $\varphi'$ in place of $\varphi$
  is applicable because we have $\omega=1\in\set N$.
  The bound using $\varphi$ is $d>(21r-18)/(r-4)$. The first correctly predicted
  degree occurs at $r=14$.
  In contrast, the bound $d>(21r-23)/(r-4)$ using $\varphi'$ is tight for all $r>5$ and
  only off by one for the minimal order $r=5$.
  The situation is shown in Figure~\ref{fig:5}. On the right, we show a comparison
  of the sharp bound based on~$\varphi'$ (solid), the bound based on~$\varphi$ (dashed)
  and the bound which would be obtained by choosing $w=0$ instead of $w=\gamma-1+\phi_3$
  in the proof of Theorem~\ref{thm:1} (dotted).
\end{enumerate}
\end{example}

\begin{figure}



  \centerline{\raisebox{1em}{\rlap{(a)}}\qquad \includegraphics{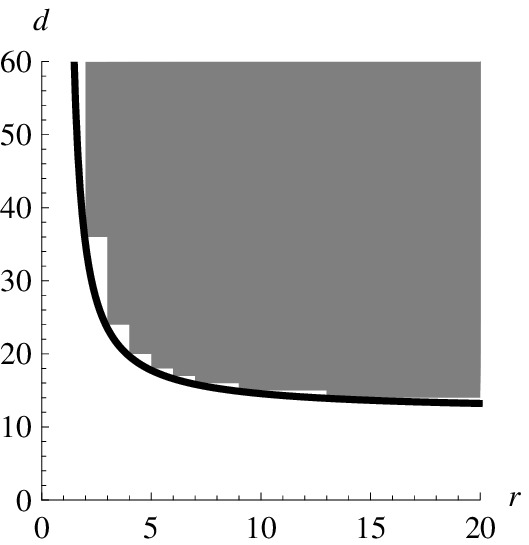}
    \hfil
    \raisebox{1em}{\rlap{(b)}}\qquad \includegraphics{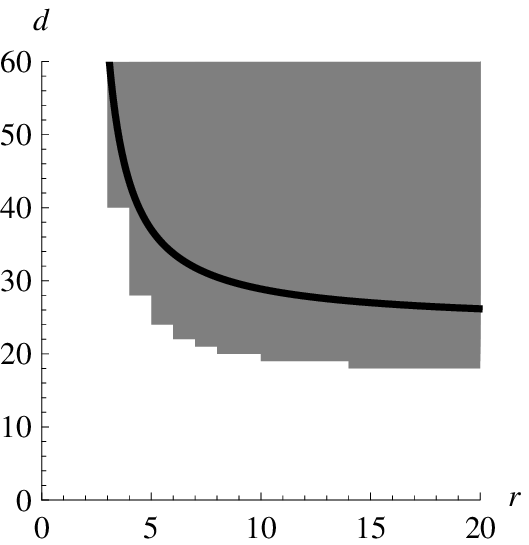}}

  \medskip



  \centerline{\raisebox{1em}{\rlap{(c)}}\qquad \includegraphics{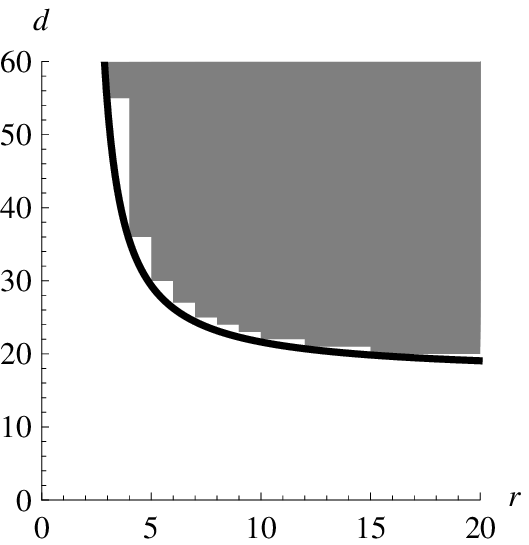}
    \hfil
    \raisebox{1em}{\rlap{(d)}}\qquad \includegraphics{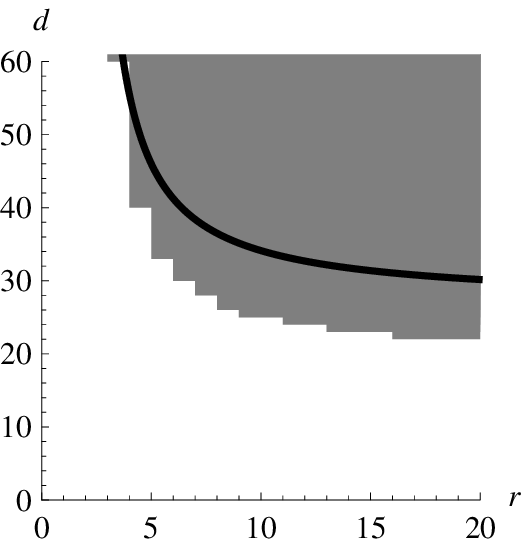}}

  \smallskip

 \caption{Sizes $(r,d)$ of creative telescoping relations together with the curve predicted
   by Theorem~\ref{thm:1}, for the hyperexponential terms discussed in Example~\ref{ex:13}.}
 \label{fig:4}
\end{figure}

\begin{figure}



  \centerline{\null\qquad\includegraphics{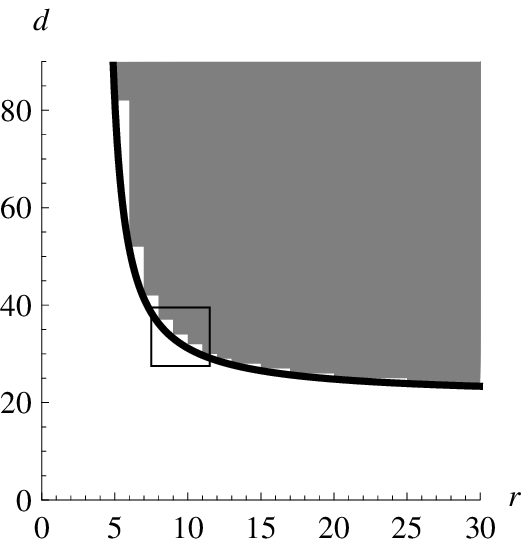}\hfil\qquad\includegraphics{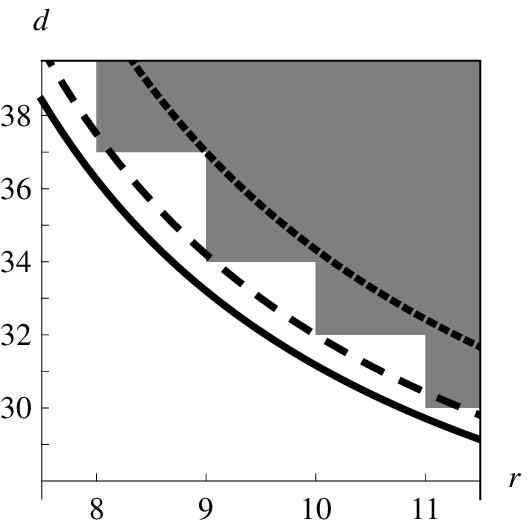}}

  \caption{Left: Sizes $(r,d)$ of creative telescoping relations together with the curve predicted
    by Theorem~\ref{thm:1}, for the term discussed in Example~\ref{ex:13}.(\ref{ex:13:5}).
    Right: a detail of the figure on the left in a larger scale, together with the curve
    based on $\varphi$ instead of $\varphi'$ (dashed) and the curve based on $w=0$ (dotted).
    The correct degrees are precisely the smallest integers strictly above the solid curve.
    The two variations both overshoot for all the points in this range.}\label{fig:5}
\end{figure}

There are several ways of refining the ansatz for $P$ and $Q$ even further in order
to achieve better estimates where ours are not sharp. Here are some ideas.

\begin{itemize}
\item The possibility of introducing extra variables without increasing the
  number of equations (depicted by the white bullets in Figures \ref{fig:1}
  and~\ref{fig:2}) rests on the observation made in Lemma~\ref{lemma:8} that the
  leading coefficients $\lc_x N_{r,i}$ are $\K$-multiples of each other, i.e.,
  that these leading coefficients generate a linear subspace of $\K[y]$ of
  dimension one. Experiments suggest that this observation can be generalized
  to the coefficients of lower degree as follows: If $V_j\subseteq\K[y]$ denotes
  the vector space generated by the coefficients of $x^{\deg_x N_{r,i} - j}$ in
  $N_{r,i}$ ($i=0,\dots,r$), then $V_0\subseteq V_1\subseteq\cdots\subseteq V_j$
  and $\dim V_j\leq j+1$ at least for small~$j$. If this is true, it would allow
  adding more extra variables without increasing the number of equations.
\item In general, comparing coefficients of the monomials $x^i y^j$ of a
  polynomial $S$ to zero results in a linear system with $(\deg_x S+1)(\deg_y
  S+1)$ equations. But if $S$ contains some factor which is free of the
  variables $p_{i,j}$ and~$q_{i,j}$, then canceling this factor before
  comparing coefficients results in a system with fewer equations and the same
  number of variables. While in our case, it is too much to hope for a factor
  which would divide $S$ as a whole, it seems that at least in some cases,
  factors can be removed from $\lc_x S \in\K[y]$ or $\lc_y S\in\K[x]$. For
  example, when $\deg_x a>\deg_x b$ and $\deg_y a>\deg_y b$, it can be shown
  that $\prod_{\ell=1}^L\lc_xc_\ell\bigm| \lc_x S$ and
  $\prod_{\ell=1}^L\lc_yc_\ell\bigm| \lc_y S$, so
  $\sum_{\ell=1}^L\bigl(\deg_y\lc_xc_\ell + \deg_x\lc_yc_\ell\bigr)$ equations can be discarded
  in this case.
\end{itemize}

We have not worked out the influence of these variations in full generality, but
only on some examples. It turned out that they indeed lead to tighter estimates,
but the difference is rather small, and decays to zero for large~$r$. At the
same time, they would lead to much more complicated formulas. We do not know the
reason for the gap in Examples~\ref{ex:13}.(\ref{ex:13:2})
and~\ref{ex:13}.(\ref{ex:13:4}) between the curve from Theorem~\ref{thm:1} and
the boundary of the gray region for $r\to\infty$. Even though it appears more
important for a bound to be tight for small orders than for large ones, we would
be very interested in seeing a refined bound which closes this gap.

It is also interesting to compare the gray regions for hyperexponential terms composed
from dense random polynomials with the gray regions for hyperexponential terms of the same
shape that originate from some specific application. According to our experiments, the
shape of the gray region for a randomly chosen term $h=c_0\exp(a/b)\prod_{\ell=1}^L c_\ell^{e_\ell}$
only depends on the number $L$ of factors in the product, the degrees of the polynomials
$a,b,c_0,\dots,c_L$, and the exponents $e_1,\dots,e_L$.
However, input containing sparse polynomials or polynomials which in some other sense have
a ``structure'' may well have considerably smaller degrees.

\begin{figure}


  \centerline{\includegraphics{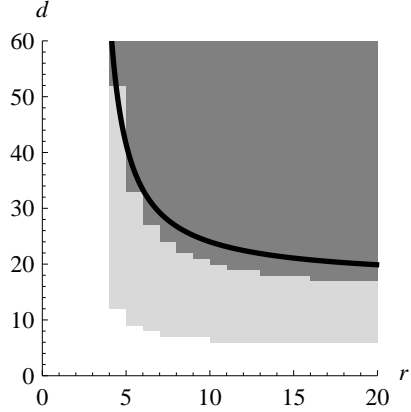}}

  \caption{Gray regions for the two terms $h$ (light gray) and $g$ (dark gray) from Example~\ref{ex:16}.
    Although all Greek parameters have the same values for $h$ and $g$ (and hence, Theorem~\ref{thm:1}
    gives the same degree estimation curve), the actual gray regions differ significantly.}\label{fig:7}
\end{figure}

\begin{example}\label{ex:16}
  If $a_{n,k}$ denotes the number of HC-polynomioes with $n$ cells and $k$ rows
  \cite[Section~4.9]{wilf89},
  then
  \[
    \sum_{n,k=0}^\infty a_{n,k}x^n y^k=\frac{x y(1-x)^3}{(1-x)^4 - x y (1-x-x^2+x^3+x^2 y)}.
  \]
  A differential equation for the generating function $\sum_{n=0}^\infty a_{n,n}x^n$ of the
  number of HC-polynomials with $n$ cells and $n$ rows can be obtained by applying creative telescoping
  to the rational function obtained from the rational function above by substituting $x$ by~$y$, $y$ by~$x/y$,
  and dividing the result by~$y$. Let thus
  \[
    h = \frac1 y\frac{y \frac xy(1-y)^3}{(1-y)^4 - y \frac xy (1-y-y^2+y^3+y^2 \frac xy)}
      = \frac{x(1-y)^3}{y((1-y)^4 - x(1-y+x y-y^2+y^3))}.
  \]
  Here we have $c_0=x(1-y)^3$, $a=0$, $b=1$, $c_1=y$, $c_2=((1-y)^4 - x(1-y+x y-y^2+y^3)$, $e_1=e_2=-1$.
  The gray region for $h$ is shown in light gray in Figure~\ref{fig:7}.
  For comparison, the same figure contains the gray region (in dark gray) for a term $g$ which
  was obtained from~$h$ by replacing $c_0$ and $c_2$ by dense random polynomials with $\deg_xc_0=1$,
  $\deg_yc_0=3$, $\deg_xc_2=2$, $\deg_yc_2=4$, so that all the Greek parameters have precisely the
  same values for $g$ and~$h$.

  Theorem~\ref{thm:1} predicts relations whenever $d\geq\frac{17r-2}{r-3}$ (black curve), which is
  a good estimate for the generic term~$g$ but a significant overestimation for the special term~$h$.
\end{example}

\section{Consequences and Applications}\label{sec:6}

Our theorem contains as a special case Theorem~cAZ of \cite{apagodu06}, which 
says that a (non-rational) hyperexponential term always admits a telescoper of order
$r=\gamma+1$, but makes no statement about its degree~$d$. Similarly,
we can also give an estimate for the possible degrees~$d$ without paying attention
to their orders~$r$.

\begin{corollary}
  \begin{enumerate}
  \item For every hyperexponential term $h$, there exists a creative telescoping
    relation of order $r=\psi+1=\gamma+1-\phi_3$.
  \item For every hyperexponential term $h$, there exists a creative telescoping
    relation of degree
    \[
      d=\vartheta+1=\left\{\begin{array}{ll}
          (\alpha+\beta)(2\gamma-1+\phi_3)+\gamma&\quad\text{if $\deg_xa>\deg_xb$;}\\
          \alpha(2\gamma-1+\phi_3)&\quad\text{if $\deg_xa\leq\deg_xb$.}
      \end{array}\right.
    \]
  \end{enumerate}
\end{corollary}
\begin{proof}
  Both claims are immediate by the formulas given in Theorem~\ref{thm:1}.
\end{proof}

In connecting order~$r$ and degree~$d$ into a single formula,
Theorem~\ref{thm:1} makes a much stronger statement than this
corollary. Assuming for simplicity that the bounds of Theorem~\ref{thm:1} are
tight, we can use them to compute optimal choices for order and degree of the
telescoper. There are various quantities which one may want to minimize.
Besides asking for a bound on the minimal order or the minimal degree, as
carried out above, we may ask for a choice $(r,d)$ where the computational
cost is minimal, or the total size $S(r,d):=(r+1)(d+1) +
(s_1+1)(\deg_xc_0+\gamma(r-1)+2)$ of the output (consisting of telescoper and
certificate), or the size $T(r,d):=(r+1)(d+1)$ of the output
telescoper alone.  Or, if the telescoper~$P$ is to be transformed into a
recurrence for the series coefficients of its solutions, one may want to
minimize the order of this recurrence, which is bounded by $R(r,d):=r+d$
\citep[see, e.g., Thm. 7.1 in][]{kauers10j}.

For minimizing the computational cost, we first have to fix a particular
algorithm for computing $P$ and $Q$ for given~$h$. We are not forced to follow
the algorithm which is implicit in the analysis of Sections~\ref{sec:4}
and~\ref{sec:5} (making an ansatz, comparing coefficients with respect to $x$
and~$y$ to zero, and solving a linear system of equations over~$\K$). In fact,
this algorithm has a rather poor performance. It is much better to do a
coefficient comparison with respect to $y$ only and to solve a linear system of
equations over~$\K(x)$. This is also what is proposed in the original articles
\citep{almkvist90,mohammed05,apagodu06} and what is used in practice
\citep{koutschan09,koutschan10b}. Output sensitive linear system solvers based
on Hermite-Pad\'e approximation \citep{beckermann94,storjohann05,bostan07} are
able to determine the degree~$n$ solutions of a linear system over $\K(x)$ with
$m$ variables and at most $m$ equations using $\O^\sim(n m^3)$
operations in~$\K$. Since an ansatz over $\K(x)$ will have only $r+1$ variables
coming from the telescoper, $\deg_yc_0+\gamma(r-1)-\phi_3+2$ variables coming
from the certificate, and a solution of degree $s_1$ with respect to~$x$,
it seems reasonable to assume that the computational cost is minimal
for a choice $(r,d)$ which minimizes the function
$C(r,d):=s_1(\deg_yc_0+(\gamma+1)r-\gamma-\phi_3+3)^3$.

\begin{example}\label{ex:17}
  Consider a hyperexponential term~$h=c_0\exp(a/b)\sqrt{c_1}$ where
  $a,b,c_0,c_1\in\K[x,y]$ have the degrees
  $\deg_xa=\deg_ya=\deg_xb=\deg_yb=1$, $\deg_xc_0=\deg_yc_0=2$, $\deg_xc_1=4$,
  $\deg_yc_1=6$. We are in case~\ref{thm:1:2} of Theorem~\ref{thm:1} and have
  $\alpha=6$, $\beta=-1$, $\gamma=8$, $\omega=4$, $\delta=5$, $\phi_1=0$, $\phi_2=0$,
  $\phi_3=0$. According to the theorem, a creative telescoping relation
  exists for $(r,d)$ with $r\geq7$ and $d\geq(89r-40)/(r-6)+1=(90r-46)/(r-6)$.

  On the curve $d=(90r-46)/(r-6)$, the cost function $C(r,d)=(6r+d-16)(9r-3)^3$
  assumes its minimal value for $r=8$ rather than for the minimal order $r=7$.
  Finding this optimal value is easy: regard $r$ temporarily as real variable
  and use calculus to determine the minimum of $C(r,\tfrac{90r-46}{r-6})$. This
  gives a minimum point near~$r=7.679$. It follows that the minimum for
  $r\in\set N$ is either at $r=7$ or at $r=8$. Comparing the actual values of
  $C$ at these two points indicates that the 8th order telescoper is about 8\%
  cheaper than the 7th order operator, and hence the cheapest operator of all.

  By similar calculations, we find that the output size (telescoper and
  certificate combined) is minimized for $r=10$, the size of the telescoper
  alone is minimized for $r=12$, and the order of the recurrence associated to
  the telescoper is minimized for $r=28$. See Figure~\ref{fig:6} for an
  illustration.
\end{example}

\begin{figure}


  \centerline{\includegraphics{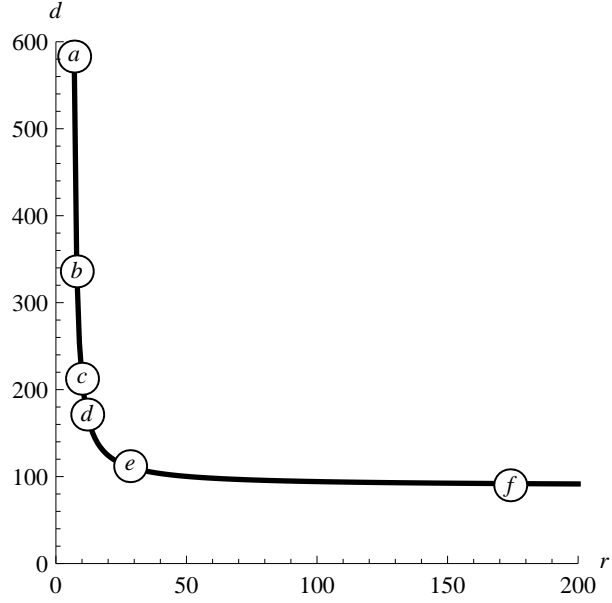}}

  \caption{Points $(r,d)$ on the curve for which
    (a) the order,
    (b) the computational cost,
    (c) the size of telescoper and certificate combined,
    (d) the size of the telescoper only,
    (e) the order of the recurrence corresponding to the telescoper, and
    (f) the degree
    is minimal.}\label{fig:6}

\end{figure}

For the moment, the term~$h$ considered in the above example is a bit too big to
actually compute the creative telescoping relations of orders 7 and 8 and
compare the difference of the timings to the predicted speedup of~8\%.  On
smaller examples, the minimal (predicted) complexity is achieved for the minimal
order operator. It may seem that an improvement by just a few percent is not
really worth the effort. But in fact, the improvement gained in the example is
just the tip of an iceberg. Asymptotically, as the input size increases, the
speedup becomes more and more significant. In the next result, which is a
generalization and a refinement of a result of~\cite{bostan10b}, we give precise
estimates.

\begin{corollary}\label{corr:1} Let $h$ be a hyperexponential term and
  $\tau=\max\{\alpha,\gamma,\deg_xc_0,\deg_yc_0\}$. 
  Let $\kappa$ be an increasing sublinear function with the property
  that degree $n$ solutions of a linear system with $m$ variables and at
  most $m$ equations over $\K(x)$ can be computed with $nm^3\kappa(\max\{n,m\})$ operations
  in~$\K$. Then a creative telescoping relation of order $r=\tau-1+\phi_3$ can be
  computed using
  \[
    2\kappa(2\tau^3)\tau^9 + \O^\sim(\tau^8)
  \]
  operations in~$\K$.
  If $r$ is chosen such that
  \[
    r = \tfrac14(1+\sqrt{17})\tau + \O(1)
    \leq 1.281\tau + \O(1)
  \]
  then a creative telescoping relation of order~$r$ can be computed using
  \[
    \tfrac1{32}(349+85\sqrt{17})\kappa(11\tau^2)\tau^8 + \O(\tau^7)
    \leq 21.86\kappa(11\tau^2)\tau^8 + \O^\sim(\tau^7)
  \]
  operations in~$\K$.
  In particular, creative telescoping relations for hyperexponential terms can be
  computed in polynomial time.
\end{corollary}
\begin{proof}
  First assume $\deg_xa>\deg_xb$. According to Theorem~\ref{thm:1}, there exists a creative
  telescoping relation of order $r$ and degree~$d$ whenever $r\geq\tau-1+\phi_3$ and
  \[
    d\geq f(r):=\frac{(2\tau^2+(2\beta+\phi_3)\tau+(\phi_3-1)\beta)r + \O(\tau^2)}{r-\tau+2-\phi_3},
  \]
  where the term $\O(\tau^2)$ is independent of~$r$.
  A creative telescoping relation of order $r$ and degree~$d$ can be computed using
  at most
  \[
    C(r,d)=\bigl((r+1)\tau+3-\phi_3\bigl)^3
           \bigl((\beta+\tau)r+d-\beta(\tau+\phi_3)-\phi_2-1\bigr)
           \kappa\bigl((\beta+\tau)(r+1)+d\bigr)
  \]
  operations in~$\K$. The claim follows from evaluating $C(r,f(r))$ at $r=\tau-1+\phi_3$
  and $r=\tfrac14(1+\sqrt{17})\tau+\O(1)$, respectively, and replacing the 
    arguments of $\kappa$ by generous upper bounds.

  For the case $\deg_xa\leq\deg_xb$, the estimates are proved analogously. Although the
  formulas for $f(r)$ and $C(r,d)$ are slightly different in this case, the final result
  turns out to be the same. We leave the details to the reader.
\end{proof}

The strange constant $\tfrac14(1+\sqrt{17})$ in Corollary~\ref{corr:1} is chosen
such as to minimize the multiplicative constant in the complexity bound under
the simplifying assumption that $\kappa$ is constant. It was
determined by first equating $\frac{d}{dr}C(r,f(r))$ to zero, which yielded the
optimal choice of~$r$ as an algebraic function in $\tau$, $\beta$,
and~$\phi_3$. The term $\tfrac14(1+\sqrt{17})\tau$ is the dominant term in the
asymptotic expansion of this function for $\tau\to\infty$. It is perhaps
noteworthy that the choice of the constant is irrelevant for achieving a cost of
$\O^\sim(\tau^8)$, as long as the constant is greater than~$1$. Taking $r=u\tau$ for
arbitrary but fixed $u>1$ leads to the complexity bound
$\frac{u^4(u+1)}{u-1}\kappa\tau^8+\O^\sim(\tau^7)$.  The choice
$u=\tfrac14(1+\sqrt{17})$ only minimizes the leading coefficient. Since
$\tfrac14(1+\sqrt{17})\approx1.28$, the result indicates that when $\alpha$ and
$\gamma$ are large and approximately equal, it appears to be most efficient to
compute a telescoper whose order is about 30\% larger than the minimum order.

In the same way as exemplified in Corollary~\ref{corr:1}, we have also determined the
choices for~$r$ for which some other quantities become minimal.
The results are given in Table~\ref{tab:1}.

\begin{table}
\begin{center}
\begin{tabular}{cc|c|c|c|c|c}
    & $r$ & $C(r,d)$ & $S(r,d)$ & $T(r,d)$ & $R(r,d)$ & $d$ \\\hline
(a) & $\tau$ & $2\kappa\tau^9$ & $\frac{2}{2-\phi_3}\tau^5$ & $2\tau^4$ & $2\tau^3$ & $2\tau^3$ \\
(b) & $\frac{1+\sqrt{17}}4\tau$ & $\frac{349+85\sqrt{17}}{32}\kappa\tau^8$
    & $\frac{53+13\sqrt{17}}8\tau^4$
    & $\frac{11+3\sqrt{17}}2\tau^3$ & $(5+\sqrt{17}) \tau^2$ & $(5+\sqrt{17}) \tau^2$ \\
(c) & $\frac{1+\sqrt5}2\tau$ & $\frac{29+13\sqrt5}2\kappa\tau^8$ & $\tfrac{11+5\sqrt5}2\tau^4$
    & $(4+2\sqrt5)\tau^3$ & $(3+\sqrt5)\tau^2$ & $(3+\sqrt5)\tau^2$ \\
(d) & $2\tau$ & $48\kappa\tau^8$ & $12\tau^4$ & $8\tau^3$ & $4\tau^2$ & $4\tau^2$ \\
(e) & $\sqrt2\tau^{3/2}$ & $4\kappa\tau^{10}$ & $2\tau^5$ & $2\sqrt2\tau^{7/2}$ & $2\tau^2$ & $2\tau^2$ \\
(f) & $2\tau^3$ & $16\kappa\tau^{16}$ & $4\tau^8$ & $4\tau^5$ & $2\tau^3$ & $2\tau^2$
\end{tabular}
\end{center}

\medskip

\caption{Minimizing various functions on the curve of Theorem~\ref{thm:1}.
  The table shows the order~$r$, the complexity~$C(r,d)$, the output size~$S(r,d)$
  of telescoper and certificate, the output size~$T(r,d)$ of the telescoper only,
  the recurrence order~$R(r,d)$, and the degree~$d$ of the
  telescoper when $r$ is chosen such that
  (a)~$r$ is minimal,
  (b)~$C(r,d)$ is minimal,
  (c)~$S(r,d)$ is minimal,
  (d)~$T(r,d)$ is minimal,
  (e)~$R(r,d)$ is minimal,
  (f)~$d$ is minimal.
  The parameters $\tau$ and $\kappa$ have the same meaning as in Corollary~\ref{corr:1}.
  The arguments of $\kappa$ are suppressed.
  Only the dominant terms of the asymptotic expansion for $\tau\to\infty$ are given.
  In rows (e) and~(f), the values for $d$ differ only in the lower order terms.}\label{tab:1}
\end{table}

As a final application, we improve some of the results given by \cite{bostan07}
on differential and recurrence equations related to algebraic functions. Let
$m\in\K[x,y]$ be irreducible with $\deg_ym\geq1$, and let $a\in\K[[x]]$ be such
that $m(x,a(x))=0$. According to Proposition~2 in their paper, if $P+D_yQ$
is a creative telescoping relation for~$y(D_ym)/m$, then $Pa=0$.  Thus we can
use our results about creative telescoping to derive estimates for differential
equations for~$a$.

\begin{corollary}
  Let $m\in\K[x,y]$ and $a=\sum_{n=0}^\infty a_nx^n\in\K[[x]]$ be as above and
  write $\tau_x:=\deg_xm$, $\tau_y:=\deg_ym$.
  Assume $\tau_x>0$ and $\tau_y>0$. Then
  \begin{enumerate}
  \item\label{cor:17:1} The series $a$ satisfies a linear differential equation of order $r=\tau_y$ with
    coefficients of degree
    \[
      d=2\tau_x\tau_y^2-\tfrac12\tau_y^2+\tau_x\tau_y-\tfrac32\tau_y+\tau_x+3.
    \]
  \item\label{cor:17:2} The series $a$ also satisfies a linear differential equation of order $r=2\tau_y$
    with coefficients of degree
    \[
      d= 4\tau_x\tau_y -\tfrac12\tau_y - 3\tau_x - 1 + \Bigl\lceil4\frac{\tau_x+1}{\tau_y+1}\Bigr\rceil.
    \]
  \item\label{cor:17:3} The coefficient sequence $(a_n)_{n=0}^\infty$ satisfies a
    linear recurrence equation of order
    \[
      \Bigl\lceil 2\tau_x\tau_y +\tau_y-1+\sqrt{(8\tau_y^2-4\tau_y+4)\tau_x - 2\tau_y^2 -6\tau_y+12}\,\Bigr\rceil
    \]
    with polynomial coefficients of degree
    \[
      \Bigl\lceil\tau_y-1+\tfrac12\sqrt{(8\tau_y^2-4\tau_y+4)\tau_x - 2\tau_y^2 -6\tau_y+12}
      \,\Bigr\rceil.
    \]
  \end{enumerate}
\end{corollary}
\begin{proof}
  For $h=y(D_ym)/m$ we have $\deg_xc_0\leq\alpha=\tau_x$, $\deg_yc_0=\gamma=\tau_y$,
  $\omega\leq0$, $\delta\leq1$, and $\phi_3=1$.
  According to Theorem~\ref{thm:1}.(\ref{thm:1:2}),
  a creative telescoping relation of order $r$ and degree $d$ exists provided that
  $r\geq\tau_y$ and
  \[
    d\geq\frac{4\tau_x\tau_y r + 2\tau_x\tau_y-\tau_y^2-3\tau_y+2\tau_x+6}{2(r-\tau_y+1)}.
  \]
  Parts \ref{cor:17:1} and \ref{cor:17:2} follow from here by setting $r=\tau_y$
  or $r=2\tau_y$, respectively.
  For part~\ref{cor:17:3}, observe first that there exists a creative telescoping
  relation of order~$r$ and degree~$d$ where
  \begin{alignat*}1
    r&\geq\tau_y-1+\tfrac12\sqrt{(8\tau_y^2-4\tau_y+4)\tau_x - 2\tau_y^2 -6\tau_y+12},\\
    d&\geq2\tau_x\tau_y + \tfrac12\sqrt{(8\tau_y^2-4\tau_y+4)\tau_x - 2\tau_y^2
      -6\tau_y+12}.
  \end{alignat*}
 {}From here the claim follows by the fact that when a power series
  $a$ satisfies a linear differential equation of order~$r$ and degree~$d$, then
  its coefficient sequence satisfies a linear recurrence equation of order $r+d$
  and degree~$r$.
\end{proof}

These results are to be compared with the corresponding results of Bostan et al.\
(degree $4\tau_x\tau_y^2+{}$smaller terms for part~\ref{cor:17:1},
order $6\tau_y$ and degree $3\tau_x\tau_y$ for part~\ref{cor:17:2},
and order and degree $2\tau_x\tau_y+\tau_y+1$ for part~\ref{cor:17:3}), as well as with the conjectures
about the minimal sizes they found experimentally
($2\tau^3-3\tau^2+3\tau$ for part~\ref{cor:17:1} when $\tau_x=\tau_y=:\tau$
and order and degree $2\tau_x\tau_y-2-(\tau_x-\tau_y)$ for part~\ref{cor:17:3} if $\tau_y>1$).

\section{Conclusion}

What is the shape of the gray region? Where does it come from? And how can it be
exploited?---These were the guiding questions for the work described in this
article. As a main result, we have given in Theorem~\ref{thm:1} a simple
rational function whose graph passes approximately along the boundary of the
gray region, in some examples more accurately than in others. This curve was
derived from a somewhat technical analysis of the linear systems resulting from
a specific ansatz over~$\K$. Where the curve does not describe the gray region
accurately, these linear systems have solutions despite of having more equations
than variables. Some possible reasons for this phenomenon were taken into
account in the design of the ansatz, thereby improving the accuracy of the
estimate compared to a naive approach. However, as shown in
Examples~\ref{ex:13}.(\ref{ex:13:2}) and~\ref{ex:13}.(\ref{ex:13:4}), there seem
to be further effects which sometimes cause a gap between the true degrees and
our prediction. It would be interesting to know what these effects are, and to
derive sharper estimates from them. Ultimately, it would be desirable to have a
version of Theorem~\ref{thm:1} which is generically~tight.

Tight curves allow for optimizing computational cost, output sizes, and other
measures by trading order against degree. As the degree decreases when the order
grows, it is not always optimal to compute the minimal order operator. In
Example~\ref{ex:17}, we have illustrated how the curve of Theorem~\ref{thm:1}
can be used to calculate a priori the optimal orders for several interesting
measures. Of course, if the curve is not tight, these predictions may not be
correct, but even then, at least they provide some useful orientation. Tightness
of the curve is also not required for deriving asymptotic bounds on the
complexity. As we have shown in Corollary~\ref{corr:1}, the difference between
the optimal choice and other choices is significant for asymptotically large input
size. We believe that this result is not only of theoretical interest. Even if
the minimal cost may be achieved for the minimal order in any example which is
feasible with currently available hardware, it can be seen from
Example~\ref{ex:17} that it already starts to make a difference for inputs which
are only slightly beyond the capability of today's computers. We therefore expect
that the technique of trading order for degree will help to optimize the
performance of efficient implementations of creative telescoping in the near
future.

\bigskip
\noindent
\textbf{Acknowledgements.} We wish to thank Christoph Koutschan and Carsten Schneider
for valuable remarks on an earlier draft of this article.

\bibliographystyle{elsart-harv}
\bibliography{bib}

\end{document}